\documentclass[titlepage]{article}
\usepackage[]{graphicx}
\usepackage[]{color}
\usepackage{hyperref}
\hypersetup{
  colorlinks, linkcolor=red
}
\usepackage{amsmath,amsfonts,amssymb,amsthm,epsfig,epstopdf,titling,url,array}


\usepackage[section]{placeins}
\usepackage{setspace}

\usepackage{setspace}

\theoremstyle{plain}
\newtheorem{thm}{Theorem}[section]
\newtheorem{lem}[thm]{Lemma}

\theoremstyle{definition}
\newtheorem{defn}{Definition}[section]

\theoremstyle{remark}

\makeatletter
\def\maxwidth{ %
  \ifdim\Gin@nat@width>\linewidth
    \linewidth
  \else
    \Gin@nat@width
  \fi
}
\makeatother

\definecolor{fgcolor}{rgb}{0.345, 0.345, 0.345}

\usepackage{framed}
\makeatletter
 {\par\unskip\endMakeFramed%
 \at@end@of@kframe}
\makeatother

\definecolor{shadecolor}{rgb}{.97, .97, .97}
\definecolor{messagecolor}{rgb}{0, 0, 0}
\definecolor{warningcolor}{rgb}{1, 0, 1}
\definecolor{errorcolor}{rgb}{1, 0, 0}

\usepackage{alltt}
\usepackage{enumitem}
\usepackage{amsthm}
\usepackage{amsmath}
\usepackage{amssymb}
\usepackage{amsfonts}

\usepackage[style=authoryear,maxcitenames=2, doi=false]{biblatex}

\bibliography{blipVar}

\usepackage[letterpaper, portrait, lmargin=1in, rmargin=1in,
bmargin = 1.35in, tmargin = 1.35in]{geometry}
\usepackage[english]{babel}
\usepackage{graphicx}
\usepackage{float}

\usepackage{caption}

\IfFileExists{upquote.sty}{\usepackage{upquote}}{}

\title{A Fundamental Measure of Treatment Effect Heterogeneity}
\author{Jonathan Levy, Mark van der Laan, Alan Hubbard, Romain Pirracchio}

\begin{document}

\begin{titlepage}

\maketitle
\begin{abstract}
We offer a non-parametric plug-in estimator for an important measure of treatment effect variability and provide minimum conditions under which the estimator is asymptotically efficient.  The stratum specific treatment effect function  or so-called blip function, is the average treatment effect for a randomly drawn stratum of confounders.  The mean of the blip function is the average treatment effect (ATE), whereas the variance of the blip function (VTE), the main subject of this paper, measures overall clinical effect heterogeneity, perhaps providing a strong impetus to refine treatment based on the confounders.  VTE is also an important measure for assessing reliability of the treatment for an individual.  The CV-TMLE \parencite{Zheng:2010aa} provides simultaneous plug-in estimates and inference for both ATE and VTE, guaranteeing asymptotic efficiency under one less condition than for TMLE \parencite{Laan:2006aa, Laan:2011aa}.  This condition is difficult to guarantee a priori, particularly when using highly adaptive machine learning that we need to employ in order to eliminate bias.  Even in defiance of this condition, CV-TMLE sampling distributions maintain normality, not guaranteed for TMLE, and have a lower mean squared error than their TMLE counterparts.  In addition to verifying the theoretical properties of TMLE and CV-TMLE through simulations, we point out some of the challenges in estimating VTE, which lacks double robustness and might be unavoidably biased if the true VTE is small and sample size insufficient.  We will provide an application of the estimator on a data set for treatment of acute trauma patients.  
\end{abstract}
\end{titlepage}

\newpage
\section{Introduction}
A clinician might observe highly variable results for a treatment and want to know how much of this variation is due to confounders, thus motivating more precision in how treatment is assigned.  Such variation also provides a measure of what to expect from treatment on an individual level in beyond the average treatment effect (ATE).   The stratum specific treatment effect function or so-called blip function, is defined as the average treatment effect for a randomly drawn stratum of confounders.  We employ targeted learning \parencite{Laan:2011aa} to construct simultaneous plug-in estimators of ATE and VTE, amounting to the sample mean and variance of blip function estimates.  Under TMLE conditions, standard error estimates obtained by computing the standard deviation of the efficient influence curve approximation are asymptotically as small as any regular asymptotically linear estimator and nominally cover the truth.  Without the targeting step in our estimator (see section 2.3), employing the non-parametric bootstrap might not guarantee valid inference \parencite{Vaart:1996aa} for our plug-in estimates if we wish to employ a rich ensemble of machine learning algorithms in our prediction of the outcome model and possibly the treatment mechanism.  In the case of our prediction methods being costly, we also save considerable time in the process.  \\

We employ ensemble machine learning so as to break from narrow parametric model assumptions that do not respect real knowledge of the statistical model.  In doing so, we will see the CV-TMLE has an advantage over the TMLE in that it does not require a donsker condition on the initial predictions, enabling more flexibility in the ensemble learning we employ.  Estimating VTE lacks the desirable robustness properties present when estimating ATE.  Particularly, for a randomized trial, our CV-TMLE estimate for ATE is consistent where as knowledge of the treatment mechanism does not guarantee consistent VTE estimates.  The lack of robustness is due to a stubborn second order remainder term, which we discuss at length.  Despite the appeal of targeted learning as in this paper, the second order remainder can make coverage unreliable, an issue for which we offer future improvements.  We also have limitations in detecting VTE when the true VTE is relatively small for the sample size. \\

\subsection{background}
Let us define $Y_{a}$ as a random variable which, is a counterfactual outcome drawn from a population under the intervention to set treatment to $a$ as per the Neyman-Rubin potential outcomes framework \parencite{neyman1923, rubin1974}.  To obtain the variance in counterfactual treatment effect from individual to individual, one would seek $var(Y_{1}-Y_{0})$.  However, this parameter is not generally identifiable, a problem addressed by Neyman, 1923, when computing the standard errors for the counterfactual average difference in yield of two crop varietals.  Neyman realized his standard errors relied upon the correlation between counterfactual outcomes on the same plot and thus had not the data to estimate it. \\

Fisher, 1951, similarly to Neyman, realized he had difficulty directly estimating the joint distribution of $Y_1$ and $Y_0$ and suggested pairing a treated and untreated varietal in the same pot might closely resemble a set of counterfactuals differences, $Y_1 - Y_0$, from which one might apply a t-test.  Cox, 1958, facing the same issue, assumes $var(Y_1 - Y_0)=0$ for predefined homogeneous subgroups, which is very difficult to verify.  Strong assumptions must be made to identify $var(Y_{1}-Y_{0})$ in the event it is not 0, such as assuming quantiles are preserved between the counterfactual outcomes \parencite{heckman}.  Heckman, 1997 mentions combining the results of Cambanis, 1976 and Frechet, 1951 to tail bound $var(Y_{1}-Y_{0})=var(Y_{1}) + var(Y_{0}) - 2\gamma_{Y_1, Y_0} var(Y_{1})var(Y_{0})$ and use the bootstrap to test if the lower bound of the confidence interval includes a variance of 0.  For randomized trial data, Ding, Feller et al., 2016, construct a Fisher randomization test of the null hypothesis that $var(Y_1-Y_0)=0$ under the untestable assumption that there exists a universal $\tau$ so that $Y_1=Y_0+\tau$.  However, a null hypothesis of 0 is not helpful when it comes to assigning treatment based on confounders. \nocite{cox1958, frechet1951, cambanis, fisher1951, neyman1923}\\

The assumption of Cox, 1958, amounts to assuming $var(Y_1 - Y_0) = var\left(\mathbb{E}[(Y_1-Y_0)\vert W]\right)$, where $W$ is an a priori known homogeneous subgroup.  Thus he assumes all variation of $Y_1 - Y_0$ is due to confounders, the exact variation we aim to capture by estimating $var\left(\mathbb{E}[(Y_1-Y_0)\vert W]\right)$, the variance of the blip function or VTE.   In our case we need not consider a priori any homogeneous subgroups.  As a simple case to give intuition as to what VTE captures, consider $W = $ indicator of male or female, and binary outcome indicating survival if the outcome is 1. Suppose the men have a blip value of $\mathbb{E}[(Y_1-Y_0)\vert W = male] = -0.3$ and the females, $\mathbb{E}[(Y_1-Y_0)\vert W = female] = 0.7$.  Assuming men and women are of equal proportion for the population at hand, then the VTE is 0.25 and ATE is 0.2.  This would mean the patient gains from treatment an average 20\% with a standard deviation of 50\%.  One should be reminded that the VTE gives a more personal measure of what to expect from treatment, but not an individual effect variance.  For instance, within the male subgroup one might have a high or low varying random variable, $(Y_1 - Y_0 \mid male)$ and such does not count toward the VTE.  Hence a clinician's perception of highly varying outcomes does not mean the VTE is high.  Rather one would want to estimate VTE to see if the varying outcomes were due to lack of precision in applying the treatment. \\

Estimating VTE with a plug-in estimator naturally depends on using an estimate of the outcome model, $\mathbb{E}[Y_a \vert W]$,  from which to estimate the blip function, $\mathbb{E}[Y_1 \vert W]- \mathbb{E}[Y_0 \vert W]$, and then its variance over an estimate of the distribution of $W$, in our case the empirical distribution of $W$.  If one knows the blip function, for instance, one would know an optimal dynamic rule for treatment \parencite{Luedtke:2014aa}.  One could also find subgroup specific treatment effects via the blip function.  Lu et al \parencite{Lu-Tian:2014aa} offered a way to isolate interactions of treatment with confounders in a randomized trial by transforming the predictors of a parametric model. The main idea is to form a variable, $z=2A-1$, where $A$ is the usual treatment indicator, and then put the interaction of this variable with the predictors in the outcome regression.  This enables direct estimation of the blip function from which one could obtain a point estimate of the VTE.  One could also employ recursive partitioning to divide the data into homogeneous subgroups as far as treatment effects \parencite{Athey:2016aa} as well as employ random forests \parencite{Athey:2015aa}.  We could use such subgroups to compute the VTE but as noted in \parencite{Marianne-Bitler:2014aa}, establishing too rough subgroups can miss detecting treatment effect heterogeneity.  In applying the CV-TMLE or TMLE, we also may use tree regression methods within our machine learning ensemble but we are only interested in the predictive power of these methods in eliminating second order remainder term bias, as we will discuss.    

\section{Methodology}
For simplicity we consider only binary treatment but all the analysis here-in can be extended to multinomial treatment if we define the treatment effect to be a contrast of any two of the treatment levels.  The reader may assume the outcome is either binary or a bounded continuous outcome.  For the latter, we can scale the outcome to be in $[0,1]$ via the transformation $Y_s = \frac{Y - a}{b-a}$ where $a$ and $b$ are minimum and maximum outcomes respectively, obtained from the data or known a priori.  The distribution $P$ in our model defines an outcome model with conditional mean, $\bar{Q}(A,W) = E_P[Y \mid A, W]$, a treatment mechanism, $g(A \mid W) = Pr[A\mid W]$ and $p_W$, the density of $W$.  For binary outcome, the density of $P$ can be factored $p(w,a,y) = \bar{Q}(a,w)^y(1-\bar{Q}(a,w))^{1-y}g(a \mid w)^a(1- g(a \mid w)^{1-a}p_W(w)$ and the mean of the log-likelihood loss is \begin{scriptsize}$\int L(P)(w,a,y)dP = \int \left[ -\left(ylog(\bar{Q}(a,w))+(1-y)log(1 - \bar{Q}(a,w))\right) - log\left(g(a \mid w)^a(1- g(a \mid w)^{1-a}p_W(w)\right)\right] dP$.\end{scriptsize} For continuous outcome scaled to be in [0,1], such is the mean of the so-called quasibinomial loss, also minimized at the true distribution, $P_0$ \parencite{Wedderburn, McCullagh}. The benefit of scaling a continuous outcome to $[0,1]$ and then  applying quasibinomial loss is the predictions are constrained to be within 0 and 1 when applying logistic regression \parencite{boundedoc}.  Thus, when scaled back, estimates will be within $a$ and $b$.  We scale our continuous outcomes to be between 0 and 1 and use quasibinomial loss, making the targeting portion of the CV-TMLE or TMLE procedure (section 2.3) identical for bounded continuous or binary outcomes.  After the CV-TMLE or TMLE algorithm is complete, one may convert the outcomes back to their original scale and form estimates and confidence bands, for which we offer instruction.  \\
 
\subsection{Full Data Statistical Model and the link to the Observed Data}
Our full data, including unobserved measures, is assumed to be generated according to the following structural equations \parencite{Wright, Strotz, Pearl:2000aa}.  We can assume a joint distribution, $U=(U_{W},U_{A},U_{Y})\sim P_{U}$, an unknown distribution of unmeasured variables. $X = (W,A,Y)$ are the measured variables.  In the time ordering of occurrence we have $W=f_{W}(U_{W})$ where $W$ is a vector of confounders, $A=f_{A}(U_{A},W)$, where $A$ is a binary treatment and $Y=f_{Y}(U_{Y},W,A)$, where $Y$ is the outcome, either binary or bounded continuous.  We thusly define a distribution $P_{U,X}$, via $(U,X) \sim P_{U,X}$.\\

$Y_{a}$ is a random outcome under $P_{U,X}$ where we intervene on the structural equations to set treatment to $a$, i.e. $Y_{a}=f_{Y}(U_{Y},a,W)$. The full model, $\mathcal{M}^{F}$, consists of all possible $P_{UX}$. The observed data model, $\mathcal{M}$, is linked to $\mathcal{M}^{F}$ in that we observe $X=(W,A,Y)\sim P$ where $X=(W,A,Y)$ is generated by $P_{UX}$ according to the structural equations above.  Our true observed data distribution, $P_{0}$, is an element of $\mathcal{M}$, which will be non-parametric.  In the case of a randomized trial or if we have some knowledge of the treatment mechanism, $\mathcal{M}$ is considered a semi-parametric model and we will incorporate such knowledge, which we will see is much more helpful for estimating ATE than for VTE.  

\subsubsection{Parameter of Interest and Identification}
We define the stratum-specific treatment effect function or blip function as $b_{P_{UX}}(W)  =  \mathbb{E}_{P_{UX}}[Y_{1}\vert W]-\mathbb{E}_{P_{UX}}[Y_{0}\vert W]$.  Our parameter of interest is a mapping from $\mathcal{M}^{F}$ to $R^{2}$ defined by 

$\Psi^{F}(P_{UX})=(\mathbb{E}_{P_{UX}}b_{P_{UX}}(W), var_{P_{UX}}b_{P_{UX}}(W))$.\\

We will impose the randomization assumption \parencite{Robins1986, Greenland1986}, $Y_{a}\perp A \vert W$ as well as positivity, $0 < E_P[A = a \mid W] < 1$ for all $a$ and $W$.  Defining $b_{P}(W)=\mathbb{E}_{P}[Y\vert A=1,W]-\mathbb{E}_{P}[Y\vert A=0,W]$ yields $b_{P_{UX}}(W)= b_{P}(W)$ and we can identify the parameter of interest as a mapping from the observed data model, $\mathcal{M}$, to $\mathbb{R}^{2}$ via the gcomp formula \parencite{Robins1986} $\Psi(P)=(\mathbb{E}_{P}b_{P}(W),var_{P}b_{P}(W))=\Psi(P^{F}_{UX})$, i.e., ATE and VTE respectively. 

\subsection{TMLE Conditions and Asymptotic Efficiency}
 We refer the reader to Targeted Learning Appendix \parencite{Laan:2011aa} as well as \parencite{Laan:2015aa,Laan:2015ab, Laan:2006aa} for a more detailed look at the theory of TMLE and the use of targeted learning that yields our algorithm below.  We offer the reader a brief overview in service of our estimation problem at hand.\\

The efficient influence curve at a distribution, $P$, for the parameter mapping, $\Psi$, is a function of the observed data, $O\sim P$, notated as $D^\star_{\Psi}(P)(O)$.  Its variance gives the generalized Cramer-Rao lower bound for the variance of any regular asymptotically linear estimator of $\Psi$ \parencite{Vaart:2000aa}.  We also note, in our general discussion, we consider our case of a two dimensional efficient influence curve, $D^{\star}_{\Psi}(P)=(D^{\star}_{\Psi_1}(P),D^{\star}_{\Psi_2}(P))$, where $\Psi_1$ and $\Psi_2$ are ATE and VTE respectively.  One may generalize to any finite dimension from our discussion. \\  

We will employ the notation, $P_{n}f(O)$, to be the empirical average of function, $f(\cdot)$, and $Pf(O)$ to be $\mathbb{E}_{P}f(O)$.  Define a loss function, $L(P)(O)$, which is a function of the observed data, O, and indexed at the distribution on which it is defined, $P$, such that $E_{P_0} L(P)(O)$ is minimized at the true observed data distribution, $P=P_0$. The TMLE procedure maps an initial estimate, $P_{n}^{0}\in \mathcal{M}$, of the true data generating distribution to $P_{n}^{\star}\in \mathcal{M}$ such that $P_{n}L(P_{n}^{\star})\leq P_{n}L(P_{n}^{0})$ and such that $P_{n}D^{\star}(P_{n}^{\star})=0_{2\times1}$. $P_{n}^{\star}$ is called the TMLE of the initial estimate $P_{n}^{0}$.  We can then write an expansion with second order remainder term, $R_2$, as follows: $\Psi(P_{n}^{\star})-\Psi(P_{0})=(P_{n}-P_{0})D^{\star}(P_{n}^{\star})+R_{2}(P_{n}^{\star},P_{0})$. 

\subsubsection{Conditions for Asymptotic Efficiency}
Define the norm $\Vert f \Vert_{L^{2}(P)} = \sqrt{\mathbb{E}_{P}f^{2}}$. Assume the following TMLE conditions:

\begin{enumerate}
\item
$D^{\star}_{\Psi_j}(P_{n}^{\star})$ is in a P-Donsker class for all $j$. This condition can be dropped in the case of using CV-TMLE \parencite{Zheng:2010aa}. We show the advantages to CV-TMLE in our simulations.  

\item
Second order remainder condition: $R_{2,j}(P_n^*,P_0)$ is $o_{p}(1/\sqrt{n})$ for all $j$.
\item
$D^{\star}_{\Psi_j}(P_{n}^{\star})\overset{L^{2}(P_{0})}{\longrightarrow} D^{\star}_{\Psi_j}(P_{0})$ for all $j$. 

\end{enumerate}

then $\sqrt n(\Psi(P_{n}^{\star})-\Psi(P_{0})) \overset{D}{\implies} N[0_{2\times1}, cov_{P_0}(D^{\star}_{\Psi}(P_{0})_{2\times2}]$ where 
$cov_{P_0}(D^{\star}_{\Psi}(P_{0})(O)$ is a $2\times2$ matrix in our case with the $(i,j)$ entry given as $E_{P_0} D^*_{\Psi_i}(P_0)(O)D^*_{\Psi_j}(P_0)(O)$.  The $i^{th}$ diagonal of $cov_{P_0}(D^{\star}_{\Psi}(P_{0})(O)$ is the variance of the $D^*_{\Psi_i}(P_0)$ and the limiting variance of $\sqrt{n}(\Psi_i(P_n^*) - \Psi_i(P_0))$ under TMLE conditions.  Thus, our plug-in TMLE estimates and CI's given by 

$$\Psi_{j}(P_{n}^{\star})\pm z_{\alpha}*\frac{\widehat{\sigma}_n(D_{j}^{\star}(P_{n}^{\star}))}{\sqrt{n}}$$ 

will be as small as possible for any regular asymptotically linear estimator at significance level, $1-\alpha$, where $Pr(\vert Z \vert \leq z_{\alpha})=\alpha$ for Z standard normal and $\widehat{\sigma}_n(D_{j}^{\star}(P_{n}^{\star}))$ is the sample standard deviation of $\{D_{j}^{\star}(P_{n}^{\star})(O_i) \mid i \in 1:n \}$ \parencite{Laan:2006aa}.  Note, that if the TMLE conditions hold for the initial estimate, $P_n^0$, then they will also hold for the updated model, $P_n^{\star}$ \parencite{Laan:2015aa}, thereby placing importance on our ensemble machine learning in constructing $P_n^0$.

\subsubsection{The Unforgiving Remainder Term in VTE Estimation}
Computation of the remainder term is in the Appendix A and is accompanied by more rigorous analysis.  Here we provide the reader with the necessary results for our discussion.  For convenience we define the true outcome model to be $\bar{Q}_0(A,W) = E_{P_0}[Y \mid A, W]$ and the true treatment mechanism as $g_0(A \mid W) = E_{P_0}[A \mid W]$.  Let $\bar{Q}_n^0$ be the initial estimate of $\bar{Q}_0$, and $g_n$ be the estimate for $g_0$.  For estimating ATE and VTE, we will fluctuate an initial outcome model fit, $\bar{Q}_n^0$ to $\bar{Q}_n^*$ but $g_n$ will not change.   The second order remainder term for VTE is:

\noindent {\scriptsize{}
\begin{eqnarray}
&  & R_{2}(P_n^*,P_{0}) = \Psi(P)-\Psi(P_{0}+P_{0}\left(D^{\star}(P)\right)\\ 
 & = & \left(\mathbb{E}_{0}b_{0}(W)-\mathbb{E}b_n^*(W)\right)^{2}\\ 
 &  + & \mathbb{E}_{0}\left[2\left(b_n^*(W)-\mathbb{E}b(W)\right)\left(\frac{g_{0}(1\vert W)-g(1\vert W)}{g(1\vert W)}\left(\bar{Q}_{0}(1,W)-\bar{Q}_n^*(1,W)\right)-\frac{g_{0}(0\vert W)-g(0\vert W)}{g(0\vert W)}\left(\bar{Q}_{0}(0,W)-\bar{Q}_n^*(0,W)\right)\right)\right]\\ 
 & - & \mathbb{E}_{0}\left(b_{0}(W)-b_n^*(W)\right)^{2}
\end{eqnarray}
}{\scriptsize \par}

where $b_n^* = \bar{Q}_n^*(1,W) - \bar{Q}_n^*(0,W)$.  Considering (2) and (3) above, we need
\begin{equation}
\Vert\bar{Q}_{n}^{*}-\bar{Q}_{0}\Vert_{L^{2}(P_{0})}\Vert g_{n}-g_{0}\Vert_{L^{2}(P_{0})}
\end{equation}
to be $o_P(n^{-0.5})$.  If the first factor is $o_P(n^{r_{\bar{Q}}})$ and the second is $o_P(n^{r_g})$, then $r_{\bar{Q}} + r_g \leq -0.5$ will satisfy the TMLE remainder term condition 2 of section 2.2.1.   It is notable the terms disappear in the case of a randomized trial where we incorporate the known $g_0$.  (5) is also a generous upper bound for the first two terms, which depend on $\int (\bar{Q}_{n}^{*}-\bar{Q}_{0})(g_{n}-g_{0}))dP_0$ because the integrand can change sign. However, (4) is not generous in this way because the integrand is a square.  Precisely, we require $\Vert \bar{Q}_n^*-\bar{Q}_0\Vert_{L^2(P_0)}$ to be $o_P(n^{-0.25})$ with no help provided by knowing the treatment mechanism.  Hence, VTE estimation is not doubly robust.  We can apply a large data adaptive ensemble of state-of-the-art machine learning algorithms to mitigate this remainder term but we still have found it can cause bias and poor coverage.   

\subsection{One-step CV-TMLE Algorithm for ATE and VTE}
The one-step TMLE algorithm constructs a parametric submodel through the initial estimate, $P_n^0$, indexed by a single dimensional parameter, $\epsilon$: $\{P_{n,\epsilon} \mid \epsilon \in [-\delta, \delta] \}$ where $P_{n,\epsilon = 0} = P_n^0$.  The construction is performed recursively in such a way that we arrive in one step at an element of the submodel, $P_n^*$, which has minimum empirical average loss of all the elements and also solves the efficient influence curve equation, $P_n D^*_\Psi (P_n^*)(O) = 0_{2\times1}$.  We introduce in Appendix C a new iterative analog to the one-step TMLE procedure which also uses one-dimensional parametric submodels, called canonical least favorable submodels \parencite{clfm}, from which one can define the universal least favorable submodel employed in the one-step TMLE algorithm. We also mention the iterative TMLE in van der Laan and Gruber, 2016, that utilizes parametric submodels of dimension the same dimension as the parameter.  It has been conjectured that the one-step TMLE may better preserve the properties of the initial fit, $P_n^0$, than the aforementioned iterative versions, thereby leading to better finite-sample behavior of the second-order remainder term $R_2(P_n^*,P_0)$ \parencite{Laan:2015ab}. If this conjecture is correct, then we would expect similar gains by using the one-step TMLE in our setting, however, in our simulations we found no appreciable differences.  Hence, we will only discuss the one-step TMLE and one-step CV-TMLE algorithms, leaving the the difference in performance between TMLE procedures as a subject for future research. \\ 

The efficient influence curve for $\Psi(P) = (\Psi_1(P), \Psi_2(P))$, i.e. ATE and VTE,  has two components given by
\begin{eqnarray*}
D^{\star}_{\Psi_1}(P)(W,A,Y) &=& \frac{2A-1}{g(A\vert W)}(Y-\bar{Q}(A,W))+b_{P}(W)-\Psi_{1}(P)\\
D^{\star}_{\Psi_2}(P)(W,A,Y) &=& 2(b_{P}(W)-\mathbb{E}_{P}b_{P})\frac{2A-1}{g(A\vert W)}(Y-\bar{Q}(A,W))+(b_{P}(W)-\mathbb{E}_{P}b_{P})^{2}-\Psi_{2}(P) 
\end{eqnarray*}
where $W$ is a possibly high dimensional set of confounders, $A$ is a binary treatment indicator and $Y$ is a binary outcome or a continuous outcome scaled between 0 and 1.  $b_{P}(W) = E_{P}[Y \mid A=1, W] - E_{P}[Y \mid A=0, W]$.  The reader may visit Appendix A for the derivation.\\

In the algorithm below, we adjusted the original CV-TMLE procedure \parencite{Zheng:2010aa} for ease of computation without losing any theoretical properties or finite sample performance.  The convenience here is that once we obtain initial estimates, there is no difference between CV-TMLE and TMLE as far as implementation is concerned.  The reader may consult the Appendix D for the difference between this procedure and the originally defined CV-TMLE \parencite{Zheng:2010aa} regarding our parameter of interest and why neither require the condition 1 in section 2.2.1.  

\subsubsection*{The "Learning" Part of Targeted Learning: Obtaining Initial Estimates}
To perform a one-step CV-TMLE we will first randomly select V folds (usually 10), consisting of V disjoint validation sets of equal size, comprising all $n$ observations and the corresponding training sets.  Each training set is the complement of the corresponding validation set so that for each fold, training set and validation set comprise all $n$ subjects.  For v in 1:V we perform the following: Using the data-adaptive ensemble machine learning package, SuperLearner (Polley, 2009), we fit the true outcome model, $\bar{Q}_0(A,W) = E_{P_0}[Y\mid A,W]$ with $\bar{Q}_{n,v}^0(A, W)$ on the training set and use the fit to make predictions on the validation set.  The V sets of validation set predictions yield one estimate of the outcome prediction for each of the $n$ subjects, denoted $\bar{Q}_n^0(A_i, W_i)$ for the $i^{th}$ subject.  In the case of an observational study, we also might use Superlearner to estimate the treatment mechanism, $E_{P_0}[A \mid W]$, with $g_{n,v}(A \mid W)$ and predict on the validation set.  The V sets of validation set predictions, yield one prediction of the treatment assignment probability for each of the $n$ subjects denoted by $g_n(A_i, W_i)$.  The reader may note that for the one-step TMLE algorithm for ATE and VTE we would have just fit the outcome model and treatment mechanism on the whole data and computed predictions $\bar{Q}_{n}^0$ and $g_{n}$ on that same data using SuperLearner (Polley, 2009).  Such is the fundamental difference between the two procedures in that CV-TMLE only uses predictions on validation sets.   

\subsubsection*{Initialize Targeting Step}
Compute the negative log-likelihood loss for our outcome predictions.  Note, $Y_i$ is the true outcome:

\[
P_{n}L(P_{n}^{0})=-\frac{1}{n}\sum_{i=1}^{n}\left[Y_{i}\text{log}\bar{Q}_{n}^{0}(A_{i},W_{i})+(1-Y_{i})\text{log}(1-\bar{Q}_{n}^{0}(A_{i},W_{i}))\right]=L_{0}\text{ our starting loss}
\]

Compute $H_{1}^{0}(A_i,W_i)=\frac{2A_i-1}{g_{n}(A_i\vert W_i)}$ and  
$H_{2}^{0}(A_i,W_i)=2\left(b_{n}^{0}(W_i)-\frac{1}{n}\sum_{i=1}^{n}b_{n}^{0}\right)\left(\frac{2A_i-1}{g_{n}(A_i\vert W_i)}\right)$ and note $H_{1}$ will stay fixed for the entire process for this parameter.  Note $b_{n}^{0}(W)=\bar{Q}_{n}^{0}(1,W)-\bar{Q}_{n}^{0}(0,W)$.  Compute $\Vert P_{n}D^{\star}_{\Psi}(P_{n}^{0})(O)\Vert_{2}$, where $\Vert \cdot \Vert_2$ is the euclidean norm. The first component of $\Vert D^*_\Psi(P_{n}^{0})(O)\Vert_{2}$, is $\frac{1}{n}\sum_{i=1}^{n}H_{1}^{0}(A_i,W_i)(Y_i - \bar{Q}_{n}^{0}(A_i,W_i))$ and the second component is $\frac{1}{n}\sum_{i=1}^{n}H_{2}^{0}(A_i,W_i)(Y_i - \bar{Q}_{n}^{0}(A_i,W_i))$.\\
  
We note that the $D^{\star}_j(P_n^0))(O_i)$ has first component, $H_{1}^{0}(A_i,W_i)(Y_i - \bar{Q}_{n}^{0}(A_i,W_i)) + b_{n}^{0}(W_1) - \frac{1}{n}\sum_{i=1}^{n} b_{n}^{0}(W_i)$ and second component $H_{2}^{0}(A_i,W_i)(Y_i - \bar{Q}_{n}^{0}(A_i,W_i)) + (b_{n}^{0}(W_1) - \frac{1}{n}\sum_{i=1}^{n} b_{n}^{0}(W_i))^2 -\frac{1}{n}\sum_{i=1}^{n} (b_{n}^{0}(W_1) - \frac{1}{n}\sum_{i=1}^{n} b_{n}^{0}(W_i))^2 $.  The reader can notice our initial estimate of the parameter is  $(\frac{1}{n}\sum_{i=1}^{n} b_{n}^{0}(W_i), \frac{1}{n}\sum_{i=1}^{n} (b_{n}^{0}(W_1) - \frac{1}{n}\sum_{i=1}^{n} b_{n}^{0}(W_i))^2)$, our sample mean and variance of our estimated TE function values.
 
\subsubsection*{The Targeting Step}
\textbf{step 2:}
If $\vert P_{n}D^{\star}_{\Psi_j}(P_{n}^{m})\vert< \hat{\sigma}_n(D^{\star}_j(P_n^m))/n$ for $j \in \{1,2\}$ then $P_{n}^{\star}=P_{n}^{m}$ and go to step 4.  $\hat{\sigma}_n(\cdot)$ denotes the sample standard deviation as in section 2.2.1.  This insures that we stop the process once the bias is second order.  Recursions after this occurs are not fruitful. If $\vert P_{n}D^{\star}_{\Psi_j}(P_{n}^{m})\vert > \hat{\sigma}/n$, then $m = m+1$ and go to step 3. 

\textbf{step 3}

Define the following recursion, using euclidean inner product notation, $\langle \cdot, \cdot \rangle_{2}$, the same as a dot product:
{\footnotesize{}
\begin{equation}
\bar{Q}_{n}^{m}(A,W)=expit\left(logit(\bar{Q}_{n}^{m-1}(A,W))-d\epsilon\biggr\langle(H_{1}^{m-1}(A,W),H_{2}^{m-1}(A,W)),\frac{P_{n}(D^*_{\Psi}(P_{n}^{m-1})(O)}{\Vert P_{n}(D^*_\Psi(P_n^{m-1}(O)\Vert_{2}}\biggr\rangle_{2}\right)
\end{equation}
} where 
\begin{itemize}[noitemsep,nolistsep]
\item
$b_{n}^{m-1}=\bar{Q}_{n}^{m-1}(1,W)-\bar{Q}_{n}^{m-1}(0,W)$
\item
$H_{1}^{m-1}(A,W)=\left(\frac{2A-1}{g_{n}(A\vert W)}\right)$
\item 
$H_{2}^{m-1}(A,W)=2\left(b_{n}^{m-1}(W)-\sum_{i=1}^nb_{n}^{m-1}(W_i)\right)\left(\frac{2A-1}{g_{n}(A\vert W)}\right)$
\item 
d$\epsilon$ is set to 0.0001 (going smaller only costs more without improving accuracy)
\end{itemize}

For the case of TMLE, this recursively defines an estimate, $\bar{Q}_{n}^{m}(A, W)$, of the true outcome model, $\bar{Q}_0(A, W) = E_{P_0}[Y \mid A, W]$. Compute $L_{m}=-\sum_{i=1}^n\left[Y_i\text{log}\bar{Q}_{n}^{m}(A_i,W_i)+(1-Y_i)\text{log}(1-\bar{Q}_{n}^{m}(A_i, W_i))\right]$.  If $L_{m}\leq L_{m-1}$ then return to step 2. Otherwise $\bar{Q}_{n}^{m}= \bar{Q}_n^*$ and continue to step 4. \\ 

\textbf{step 4}

Our estimate for ATE and VTE is $\left(\frac{1}{n}\sum_{i=1}^{n} b_n^*(W_i),\frac{1}{n}\sum_{i=1}^{n} \left(b_n^*(W_i) - \frac{1}{n}\sum_{i=1}^{n} b_n^*(W_i)\right)^2\right)$, where $b_n^*(W_i)=\bar{Q}_{n}^{\star}(1,W_1) - \bar{Q}_{n}^{\star}(0,W_i)$.  If the outcome was scaled as $\frac{Y - a}{b-a}$ (see section 2, paragraph 1), then ATE and VTE is $\left(\frac{b-a}{n}\sum_{i=1}^{n} b_n^*(W_i),\frac{(b-a)^2}{n}\sum_{i=1}^{n} \left(b_n^*(W_i) - \frac{1}{n}\sum_{i=1}^{n} b_n^*(W_i)\right)^2\right)$.

\subsubsection{Simultaneous Estimation and Confidence bounds}
We often want to provide confidence intervals that simultaneously cover all the coordinates of $\Psi(P_{0})$ at a given significance level.  
The following is an added benefit of having the efficient influence curve at hand for we can account for correlated estimates in a tighter manner than a bonferroni correction \parencite{bonferroni}.  The reader may note we can generalize this procedure to any dimension but will use dimension 2 here as that is relevant to our parameter.  After completing the above algorithm we have, $D^*_{\Psi}(P_n^*)(O_i) = (D^*_{\Psi_1}(P_n^*)(O_i), D^*_{\Psi_2}(P_n^*)(O_i))$, for each subject indexed by $i \in 1:n$.  Consider the $2$-dimensional random variable $Z_{n} = (Z_{n,1},Z_{n,2})  \sim N(0_{2\times1}, \Sigma_{n})$, defined by two by two matrix, $\Sigma_{n}$, the sample correlation matrix of $D^*_{\Psi}(P_{n}^{\star})$.  Let $q_{n,\alpha}$ be the $\alpha^{th}$ quantile of the random variable $M_n = max(\vert Z_{n,1} \vert,\vert Z_{n,2} \vert)$.   Let $Z = (Z_1,Z_2) \sim N[0_{2 \times 1}, \Sigma]$, where $\Sigma$ is the correlation matrix of $D^{\star}_\Psi(P_0)$.  Let $q_{\alpha}$ be the $\alpha^{th}$ quantile of the random variable $M = max(\vert Z_{1} \vert,\vert Z_{2} \vert)$, i.e., the $\alpha^{th}$ quantile of the random variable giving the max number of standard deviations over the coordinates of $Z$.  We monte-carlo sample 5 million draws from the random variables $M_n$ to find $q_{n,\alpha}$.  We note that 5 million is a sufficient number to guarantee very little error in finding the true $q_{n,\alpha}$.  Applying the continuous mapping theorem \parencite{Vaart:1996aa} assures us under TMLE conditions that $Z_{n,i} \pm q_\alpha$ covers $Z_i$ for all $i \in 1:2$, at $(1-\alpha)\times 100\%$. Then we can apply the extended continuous mapping theoreom \parencite{Vaart:1996aa} to assure us $q_{n,\alpha}\longrightarrow q_{\alpha}$.  $q_{n,\alpha}$ is therefore an estimate of the number of standard errors needed to simultaneously cover both true parameter values at $(1 - \alpha)\times 100 \%$.  This results in the confidence bands
\[
\hat{Psi}_{j,n}\pm q_{n,\alpha}*\frac{\widehat{\sigma}_n(D_{j}^{\star}(P_{n}^{\star}))}{\sqrt{n}}
\]

which, will asymptotically cover all coordinates, $\Psi_{j}(P_0)$ of $\Psi(P_0)$,  simultaneously at the significance level, $1-\alpha$.  The reader may note $q_{n,\alpha}$ is the same as the bonferroni correction \parencite{bonferroni} if $\Sigma_n$ is the identity matrix.  As in section 2.2.1, $\widehat{\sigma}_n(\cdot)$ is the sample standard deviation.  If the outcomes were scaled according to $Y_s = \frac{Y - a}{b-a}$, then the standard error estimates for ATE and VTE are multiplied by $b-a$ and $(b-a)^2$, respectively.  

\section{Simulations}
We performed two different kinds of simulations, the first primarily to verify the remainder conditions in the theory of TMLE (condition 2, section 2.2.1).  The rest were performed to get a sense of what might occur with real data. Inference for all TMLE's used the sample standard deviation of the efficient influence curve approximation to form confidence intervals as per section 2. For logistic regression plug-in estimators of ATE and VTE, confidence bands were formed by using the delta method and the influence curve for the beta coefficients for intercept, main terms and interactions (see appendix b for the derivation). SuperLearner initial estimates had no accompanying measure of uncertainty since there is little theory for such, even if bootstrapping \parencite{Vaart:1996aa}.  

\subsection{Simulations with Controlled Noise}

Instead of drawing $W$ then $A$ and then $Y$ under a data generating distribution and then trying to recover the truth with various predictors or SuperLearner as we do later, we directly add heteroskedastic noise to $\bar{Q}_{0}$ in such a way that the conditions of TMLE hold and then use the noisy estimate as the initial estimate in the TMLE process. This does not necessarily
match what happens in practice because the noise we add is not related to the noise in the draw of $Y$ given $A$ and $W$. However, it is a valid way to directly test the conditions of TMLE in that we can control the noise so that the TMLE conditions hold and watch the asymptotics at play. We also note that we will assume $g_0$ is known because the other second order terms for VTE, involving bias in estimating $g_0$, are dependent on double robustness in the same way as for the ATE, for which the properties of TMLE are already well-known \parencite{Laan:2006aa, Laan:2011aa}. 

\subsubsection{Simulation Set-up}

$W_{1}\sim uniform[-3,3]$, $W_{2}\sim binomial(1,.5)$, $W_{3}\sim N[0,1]$
and $W_{4}\sim N[0,1]$.  We define $g_{0}(A\vert W)=expit(.5*(-0.8*W_1+0.39*W_2+0.08*W_3-0.12*W_4-0.15))$
, which is the true density of $A$ given $W.$  We kept our propensity scores between about 0.17 and 0.83 so as to avoid poor performance from positivity violations \parencite{positivity}.   $\mathbb{E}_{0}[Y\vert A,W]=\bar{Q}_{0}(A,W)=expit(.2*(.1*A+2*A*W_1-10*A*W_2+3*A*W_3+W_1+W_2+.4*W_3+.3*W_4))\}$ which
defines the density of $Y$ given $A$ and $W$ for a binary outcome. Define the blip function as $b(W)=\mathbb{E}_{0}[Y\vert A=1,W]-\mathbb{E}_{0}[Y\vert A=0,W]$ and we have $\Psi(P_{0})=var_{0}(b(W))=0.0636$. This is a substantial
VTE to avoid getting near the parameter boundary at 0. \\ 

below we illustrate the process for one simulation.  For each sample size, n, we performed the simulation 1000 times.  We note that $rate$ is some number which we will set to less than -1/4 (-1/3 in this case) in order to satisfy TMLE conditions. 
\begin{enumerate}
\item define $bias(A,W,n)=1.5n^{rate}(-.2+1.5A+0.2W_{1}+W_{2}-AW_{3}+W_{4})$
\item define heteroskedasticity: $\sigma(A,W,n)=0.8n^{rate}\vert3.5+0.5W_{1}+0.15W_{2}+0.33W_{3}W_{4}-W_{4}\vert$
\item define $b(A,W,n,Z)=bias(A,W,n)+Z\times\sigma(A,W,n)$ where Z is standard
normal
\item draw $\{Z_{i}\}_{i=1}^{n}$ and $\{X_{i}\}_{i=1}^{n}$ each from standard
normals
\item $\bar{Q}_{n}^{0}(1,W_{i})=expit\left(logit\left(\bar{Q}_{0}(1,W_{i})\right)+b(1,W_{i},n,Z_{i})\right)$
\item $\bar{Q}_{n}^{0}(0,W_{i})=expit\left(logit\left(\bar{Q}_{0}(0,W_{i})\right)+0.5b(1,W_{i},n,Z_{i})+\sqrt{0.75}b(0,W_{i},n,X_{i})\right)$
\item $\bar{Q}_{n}^{0}(A,W)=A*\bar{Q}_{n}^{0}(1,W)+(1-A)\bar{Q}_{n}^{0}(0,W)$
\end{enumerate}
We note that we placed correlated noise on the true $\bar{Q}_{0}(1,W)$
and $\bar{Q}_{0}(0,W)$ so as to make the blip function``estimates'' of
similar noise variance as the initial ``estimates'' for $\bar{Q}_{0}(A,W)$. 
By a Taylor series expansion about the truth, it is easy to see the
above procedure will satisify the remainder term conditions of 2.2.1.  We have that $\bar{Q}_{n}^{0}(1,W)=\bar{Q}_{0}(1,W)+\bar{Q}_{0}(1,W)(1-\bar{Q}_{0}(1,W))b(1,W,n,Z)+O(b^{2}(1,W,n,Z))$
and likewise for $\bar{Q}_{n}^{0}(0,W)$ and thus trivially, $\sqrt{\mathbb{E}_{0}\left(b_{n}^{0}(W)-b_{0}(W)\right)^{2}}$
is of order $n^{rate}$ with $rate < -1/4$.  As previously mentioned, we need not worry
about any second order terms but $\mathbb{E}_{0}\left(b_{n}^{0}(W)-b_{0}(W)\right)^{2}$
because we are using the true $g_{0}$. Condition 1 of section 2.2.1 is easily satisfied and
Condition 3, the donsker condition, is satisfied since our ``estimated''
influence curve, $D^{*}(\bar{Q}_{n}^{0},g_{0}),$ depends on a fixed
function of $A$ and $W$ with the addition of independently added
random normal noise. \\

The simulation result, displayed in figure 1, is in alignment with the theory established for the TMLE estimator of VTE but how fast the asymptotics come into play is an important issue as to the relevance of the asymptotic theory.  

\begin{figure}[H]
  \centering
  \caption{}
  \includegraphics[scale=.2]{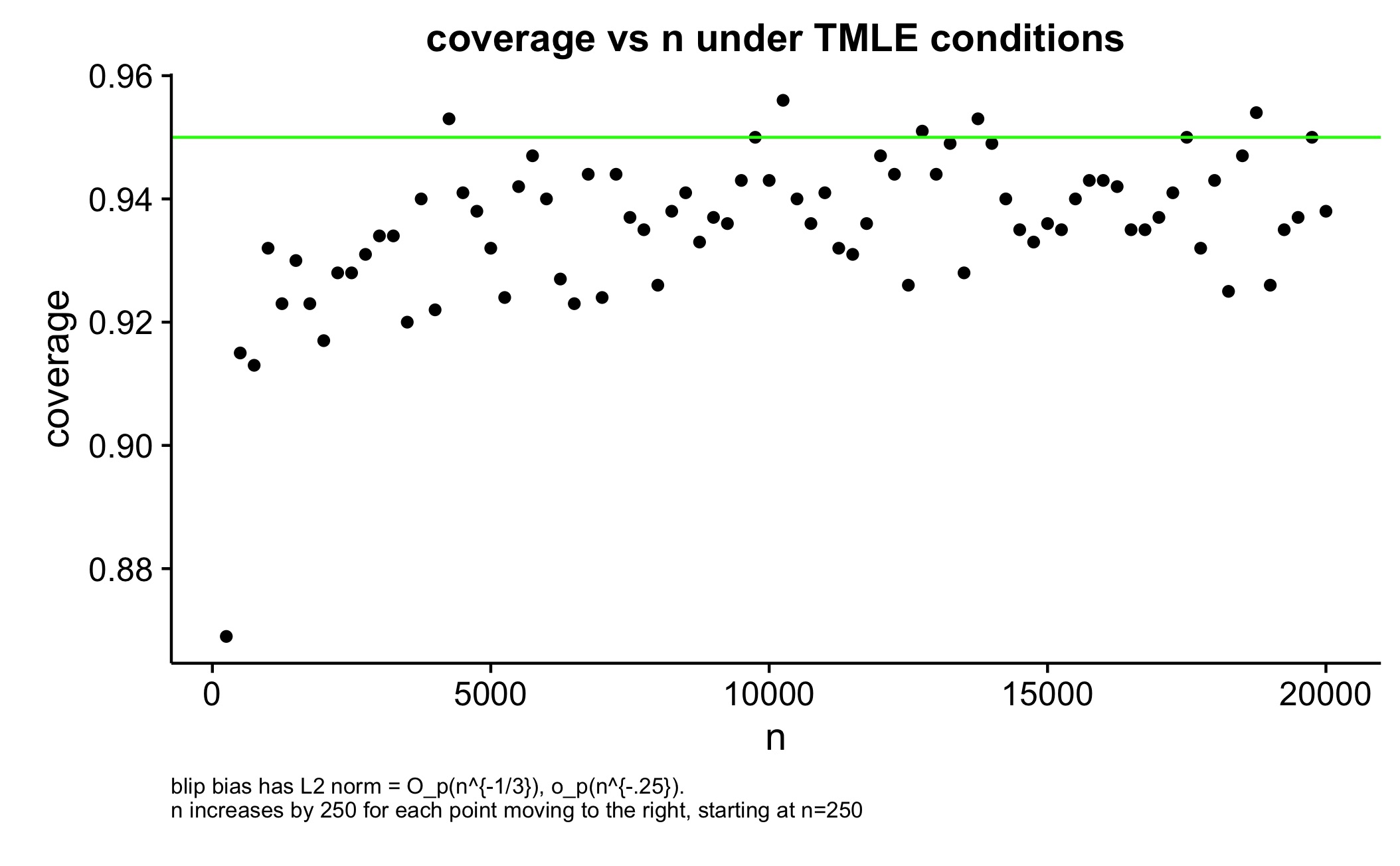}
\end{figure}

\subsection{Simulations That Are More Realistic}
We will stick with binary outcome and treatment, though the results will be comparable for continuous outcome. Unless otherwise noted, sample size n = 1000 and the number of simulations = 1000.  Throughout the simulations we generated the covariates as follows: $W_1 \sim uniform[-3,3]$, $W_2 \sim \text{standard normal}$, $W_3 \sim \text{standard normal}$ and $W_4 \sim \text{standard normal}$.  These simulations are more realistic in that we try to recover via machine learning, an "unknown" treatment mechanism and outcome model.  When we specify the models correctly we are considering a "best case" scenario where our regressions achieve parametric rates of convergence to the truth.  When we misspecify a model in the data generating system, we try to recover its non-linear functional form with ensemble machine learning, in the event that a linear model including interactions (to pick up heterogeneity) is catastrophic for estimating VTE.  If we are going to estimate VTE, a main terms linear model will assume VTE is essentially 0, so comparing ensemble learning methods with such is not very informative.

\subsubsection{Well-specified TMLE Initial Estimates, Skewing}
"Well-specified" means we are fitting a well-specified (functional form is correct for both outcome model, $\mathbb{E}[Y\vert A,W]=\bar{Q}(A,W)$ and treatment mechanism, $\mathbb{E}[A\vert W]=g_{0}(A,W)$), as in a logistic linear model for both the treatment mechanism and outcome models. The only point of these simulations is to show that TMLE preserves excellent initial estimates and also to show approximately what size sample will lead to skewing (and therefore bias) of the sampling distribution for blip variance when the truth is near the lower parameter bound of 0. We can say as a rule of thumb, a sample size of 500 or more is probably needed to even hope to get reliable estimates for blip variances in the neighborhood of 0.025 (15.8\% standard deviation), a rule confirmed by figures 2,3 and 4. $\bar{Q}(A,W) = expit(0.14 (2 A + W1 + a A W_1 - b A W_2 + W_2 - W_3 + W_4))$ for the outcome regression, varying $a$ and $b$ to adjust the size of the blip variance. $\mathbb{E}[A\vert W] = g0 = expit(-0.4 * W1 + 0.195 * W2 + 0.04 * W3 - 0.06 * W4 - 0.075)$ was the true treatment mechanism and we avoid bad positivity violations here.     


\begin{figure}[H]
  \centering
  \caption{}
  \includegraphics[scale=.2]{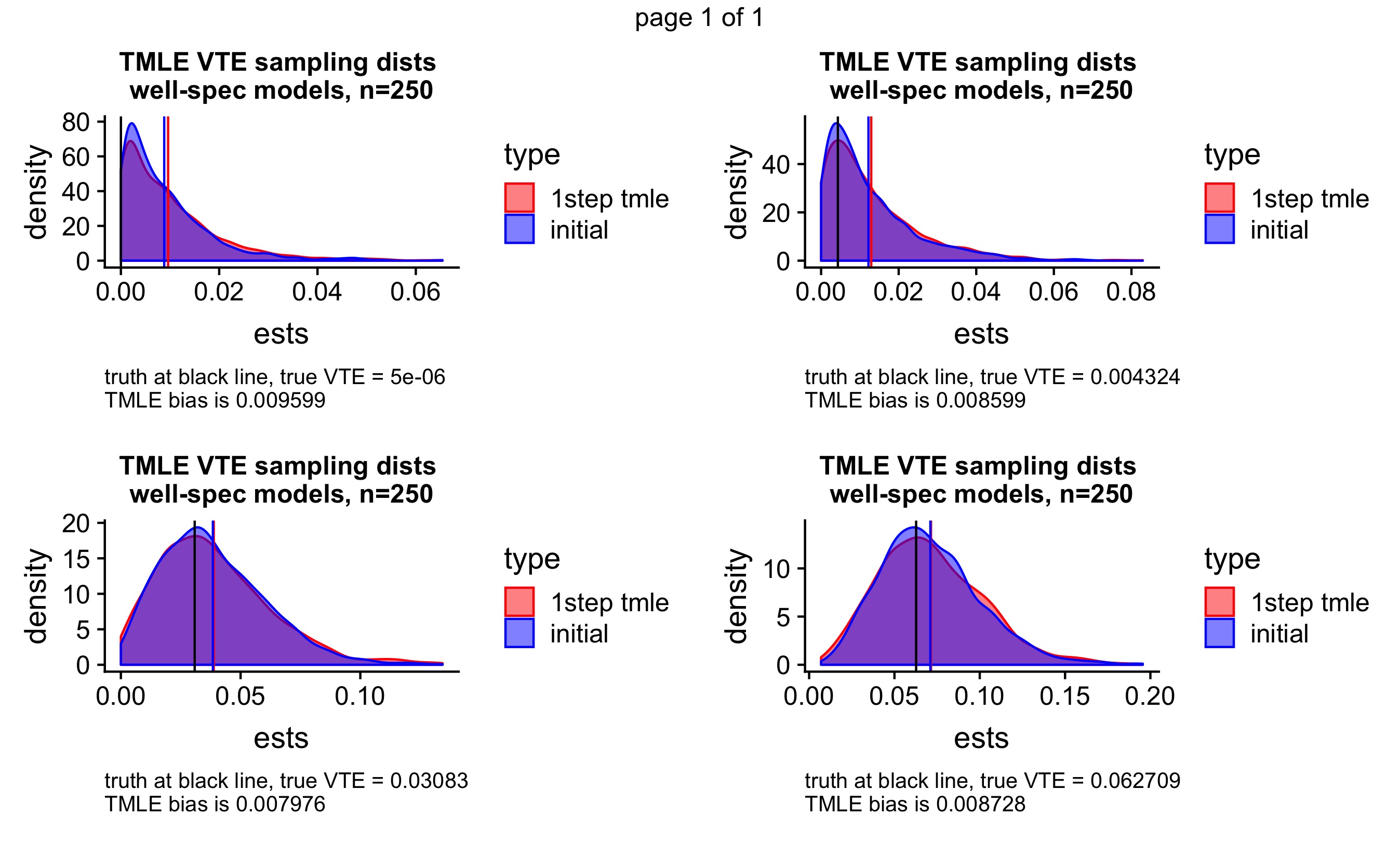}
\end{figure}

\begin{figure}[H]
  \centering
  \caption{}
  \includegraphics[scale=.2]{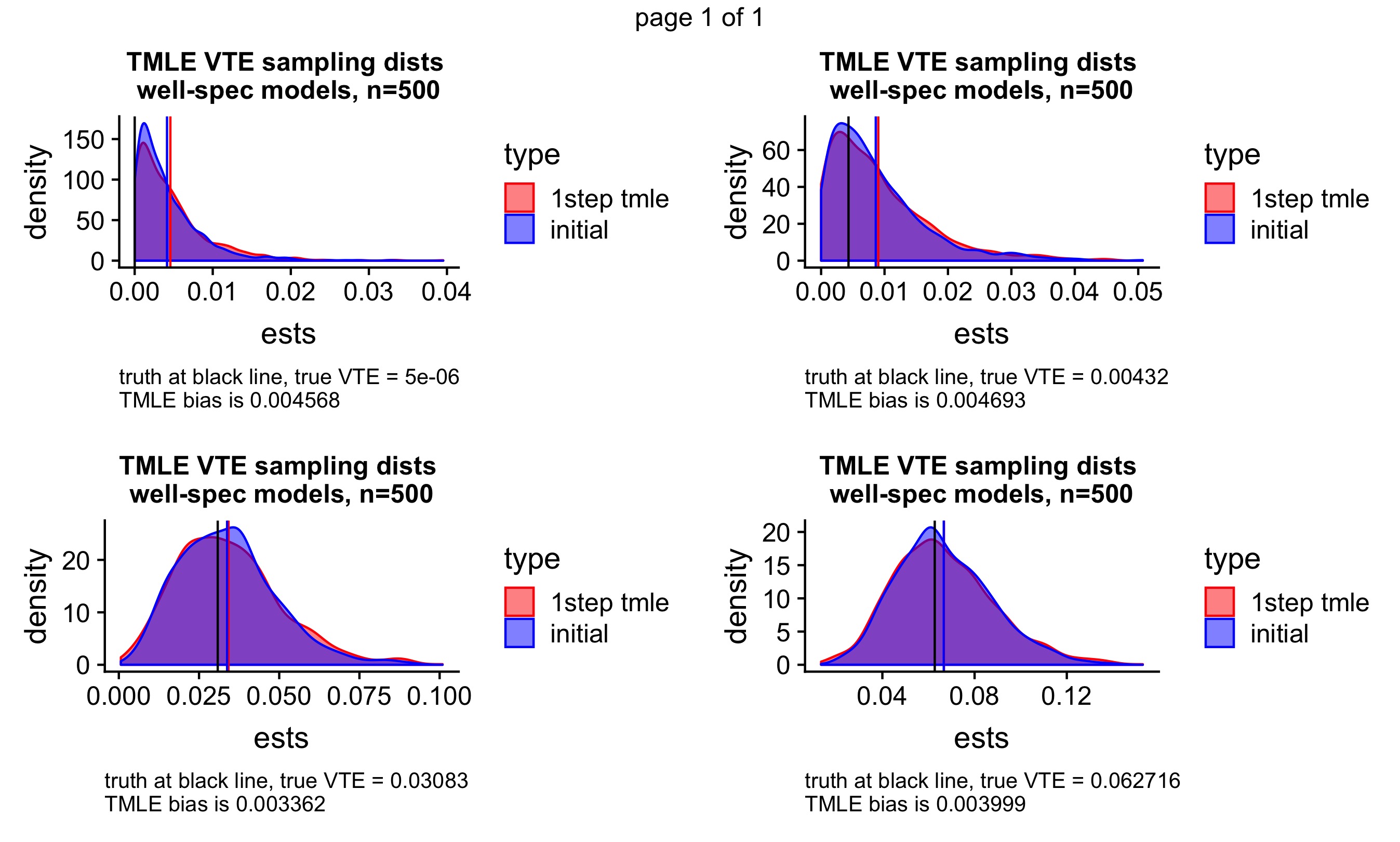}
\end{figure}

\begin{figure}[H]
  \centering
  \caption{}
  \includegraphics[scale=.2]{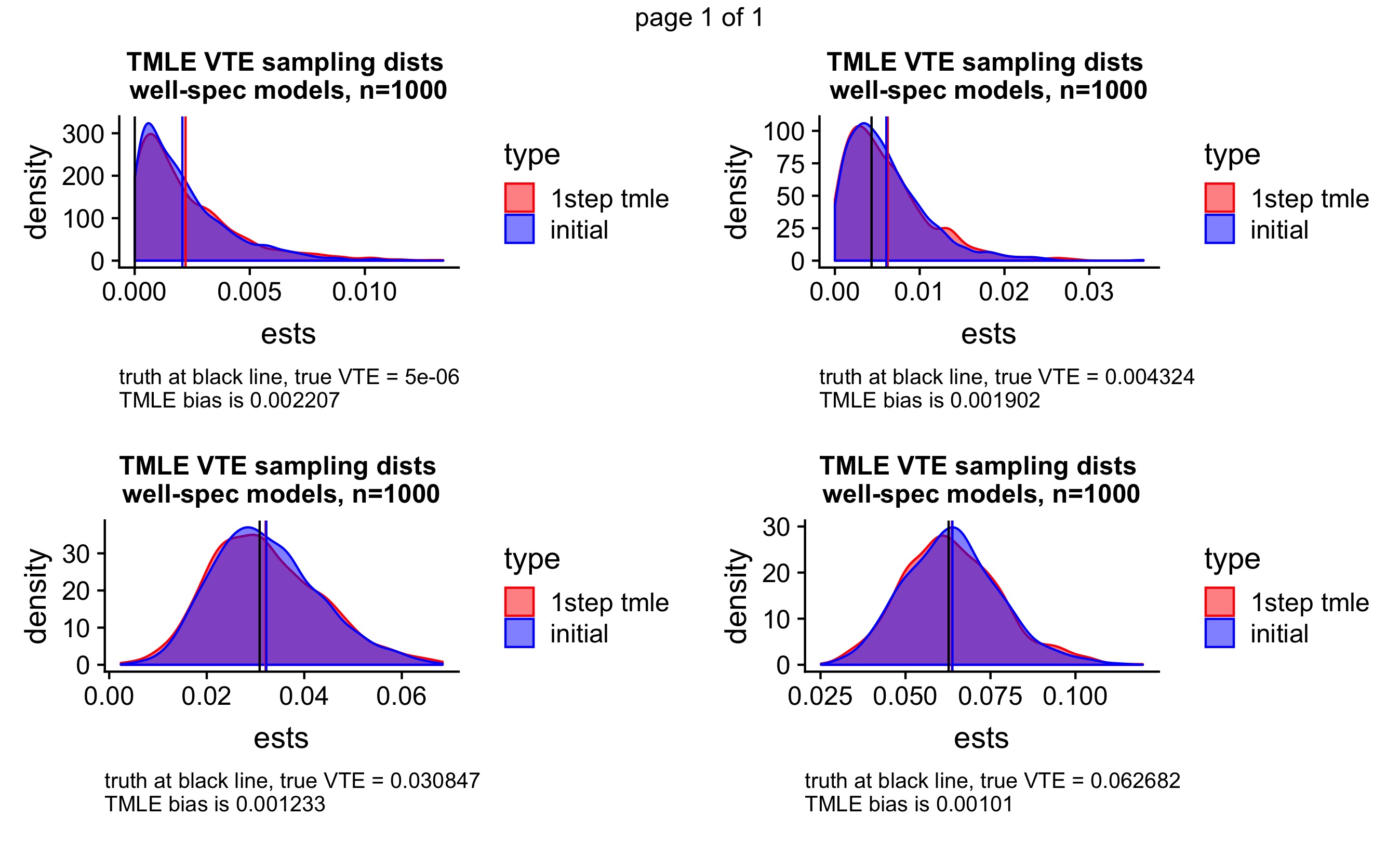}
\end{figure}

\subsubsection{SuperLearner Details For Remaining Simulations}

Targeted learning \parencite{Laan:2011aa} features the use of data adaptive prediction methods optimized by the ensemble learningR packages, such as SuperLearner \parencite{Eric-Polley:2017aa}, H2O  \parencite{LeDell:2017aa} or the most recent sl3 \parencite{sl3}.  Superlearner, which picks the best single algorithm in the library, as decided by the cross-validation of a valid loss function, has risk that converges to the oracle selector at rate $O\left(log \left(k(n)\right)/n\right)$ where $k(n)$ is the number of candidate algorithms, under very mild assumptions on the library of estimators \parencite{Mark-J.-van-der-Laan:2007aa}.  Generally the best or nearly best learner in the library is the optimal convex combination of algorithms that forms the SuperLearner predictor in our simulations and often in practice.  The SuperLearner convex combination of algorithms might be more familiar to the reader as a form of model stacking \parencite{stacking}.  

\subsubsection{Use of Highly Adaptive Lasso: Making Initial Predictions $\bar{Q}_n^0$ and $g_n$}
We refer the reader to the online supplementary materials for more complete details of the SuperLearner results. It is notable that any SuperLearner library containing highly adaptive lasso will yield asymptotically efficient estimates for TMLE, assuming the true conditional mean outcome is right-hand continuous with left-hand limits \parencite{cadlag} and has variation
norm smaller than a constant M \parencite{Laan:2015aa}.  In finite samples, however, some machine learning algorithms might be better suited for prediction and so we rely on ensemble learning.  

\subsubsection{SuperLearner Library 1, termed SL1, Avoiding Overfitting} 
This library will be indicated by "SL1" in the simulation results.
\begin{enumerate}
\item
SL.gam3, a gam \parencite{Hastie:2017aa} using degree 3 smoothing splines, screening main terms, top 10 correlated variables with the outcome and top 6.
\item
SL.glmnet\_1, SL.glmnet\_2 and SL.glmnet\_3 \parencite{Jerome-Friedman:2010aa} performed a lasso, equal mix between lasso and ridge penalty and ridge regressions.
\item
nnetMain\_screen.Main \parencite{Venbles:2002aa}
is a neural network with decay = 0.1 and size = 5 using main terms.
\item
earthMain \parencite{Mill} is data adaptive penalized regression spline fitting method. They allow for capturing the subtlety of the true functional form. We allowed degree = 2, which is interaction terms with the default penalty = 3 and a minspan = 10 (minimum observations between knots).  
\item
SL.glm \parencite{R-Core-Team:2017aa} logistic regression and we used main terms, top 6 correlated variables with outcome and top 10 as well as a standard glm with main terms and interactions (glm\_mainint\_screen.Main)
\item
SL.stepAIC \parencite{R-Core-Team:2017aa} uses Akaike criterion in forward and backward step regression
\item
SL.hal is the highly adaptive lasso \parencite{Benkeser:2016aa}, which guarantees the necessary $L_2$ rates of convergence and therefore, if included in the SuperLearner library, guarantees asymptotic efficiency \parencite{Laan:2015ab}. hal output is guaranteed to be of finite sectional variation norm and thus is guaranteed to be donsker or not overfit the data.  
\item
SL.mean returns the mean outcome for assurance against overfitting
\item
rpartPrune \parencite{Therneau:2017aa} is recursive partitioning with cp = 0.001 (must decrease the loss by this factor) minsplit = 5 (min observations to make a split), minbucket = 5 (min elements in a terminal node)
\end{enumerate}

\subsubsection{SuperLearner Library 2 termed SL2, More Aggressive, overfits a little}
This library will be indicated by "SL2" in the simulation results.
This library is identical to Library 1, except we added the following learners, which were tuned to maximize cross validated loss on a few draws from case 2a data generating distribution.  Thus these additions do not severely overfit, in general.
\begin{enumerate}
\item
SL.ranger \parencite{Wright:2017aa}: A random forest which picked 3 features at a time formed 2500 trees and had a minimum leaf size set to 10.  
\item
SL.xgboost \parencite{Chen:2017aa}: One xgboost fit on all main terms and interactions with stumps (depth 1 trees), allowing a minimum of 3 observations per node, a learning rate of 0.001 and summing 10000 trees.  We also included an xgboost using depth of 4 trees on main terms only with same shrinkage and minimum observations per node but only 2500 trees.  
\end{enumerate}

\subsubsection{Case 1: Well-Specified Treatment Mechanism, Misspecified Outcome}
The following example, encapsulated in figure 5, demonstrates three things
\begin{enumerate}
\item 
\textbf{Enormous gains possible with flexible estimation}. Using a logistic regression with main terms and interactions plug-in estimator and the delta method for inference, yielded a bias of -0.065 (the truth is 0.079), missing almost the entire blip variance and covering at 0\%. The TMLE could not help the initial estimates using the same logistic regression so reliance on a parametric model can be a disaster as opposed to ensemble learning. 
\item 
\textbf{Difference in robustness} between estimating causal risk difference and blip variance.  The severely misspecified logistic regression with main terms and interactions initial estimate for the outcome model and well-specified treatment mechanism yielded a TMLE for causal risk difference (which is doubly robust) that covered at 95.6\%,  where as for blip variance it never covers the truth.  

\item
\textbf{The advantage of CV-TMLE over TMLE}. The same SuperLearner used for initial estimates yields some skewing and bad outliers as well as bigger bias and variance for TMLE as opposed to a normally distributed CV-TMLE sampling distribution.  Just some overfitting by random forest about 20\% of the time out of the library of 18 learners managed to cause outliers for TMLE, ruining normality of the sampling distribution and causing higher bias and variance, where as CV-TMLE appeared unaffected by the overfitting.  Overfitting means essentially that the metric entropy of the class of functions considered by random forest was too big.  To give some intuition behind the donsker TMLE condition, an example of a large donkser class is the set of functions of bounded variation \parencite{Vaart:1996aa} meaning the function class is smooth in some sense, not allowing unlimited ups and downs between predictions, such as overfitting allows.  Since the influence curve approximation is defined partially in terms of the mean outcome model, overfitting causes the class of functions for the influence curve approximation to be non-donsker as well.  When trying to do a good job estimating the mean outcome model as in this simulation, CV-TMLE allows highly adaptive machine learning we need to minimize the second order remainder bias without paying a big price for overfitting.  

\end{enumerate}
\subsubsection*{Simulation Set-up}
$\mathbb{E}[Y\vert A,W] = Q0 = expit(0.28 * A + 2.8 * cos(W1) * A + cos(W1) - 0.56 * A * (W2^2) + 0.42 * cos(W4) * A + 0.14*A * W1^2)$.  $\mathbb{E}[A\vert W] = g0 = expit(-0.4 * W1 + 0.195 * W2 + 0.04 * W3 - 0.06 * W4 - 0.075)$, which we will specify model correctly in all cases with a linear logisitic fit.  True Causal Risk Difference = 0.078.  True CATE Variance = 0.085.

\begin{table}[!htbp] \centering 
  \caption{Performance of the Estimators} 
  \label{} 
\begin{tabular}{@{\extracolsep{5pt}} lcccc} 
\\[-1.8ex]\hline 
\hline \\[-1.8ex] 
 & var & bias & mse & coverage \\ 
\hline \\[-1.8ex] 
TMLE LR & $0.00001$ & $$-$0.08207$ & $0.00675$ & $0$ \\ 
LR plug-in & $0.00005$ & $$-$0.07134$ & $0.00514$ & $0$ \\ 
CV-TMLE SL2 & $0.00028$ & $$-$0.00930$ & $0.00037$ & $0.87375$ \\ 
CV-TMLE SL2\textasteriskcentered  & $0.00028$ & $$-$0.00924$ & $0.00037$ & $0.88577$ \\ 
TMLE SL2 & $0.00057$ & $0.01584$ & $0.00082$ & $0.83100$ \\ 
TMLE SL2\textasteriskcentered  & $0.00057$ & $0.01591$ & $0.00082$ & $0.86000$ \\ 
TMLE SL1 & $0.00033$ & $0.00802$ & $0.00040$ & $0.93193$ \\ 
TMLE SL1\textasteriskcentered  & $0.00033$ & $0.00804$ & $0.00040$ & $0.93994$ \\ 
\hline \\[-1.8ex] 
\multicolumn{5}{l}{* indicates causal risk difference and blip variance estimated} \\
\multicolumn{5}{l}{simultaneously with 1step tmle covering both parameters for 95\%}\\
\multicolumn{5}{l}{simultaneous confidence intervals.  LR indicates logistic regression}\\
\multicolumn{5}{l}{with main terms and interactions. SL1 Library did not overfit}\\
\multicolumn{5}{l}{out 20\% of the time with one out of 18 algorithms, still causing}\\
\multicolumn{5}{l}{outliers, necessitating CV-TMLE as a precaution.}\\
\end{tabular} 
\end{table} 


\begin{figure}[H]
  \centering
  \caption{}
  \includegraphics[scale=.2]{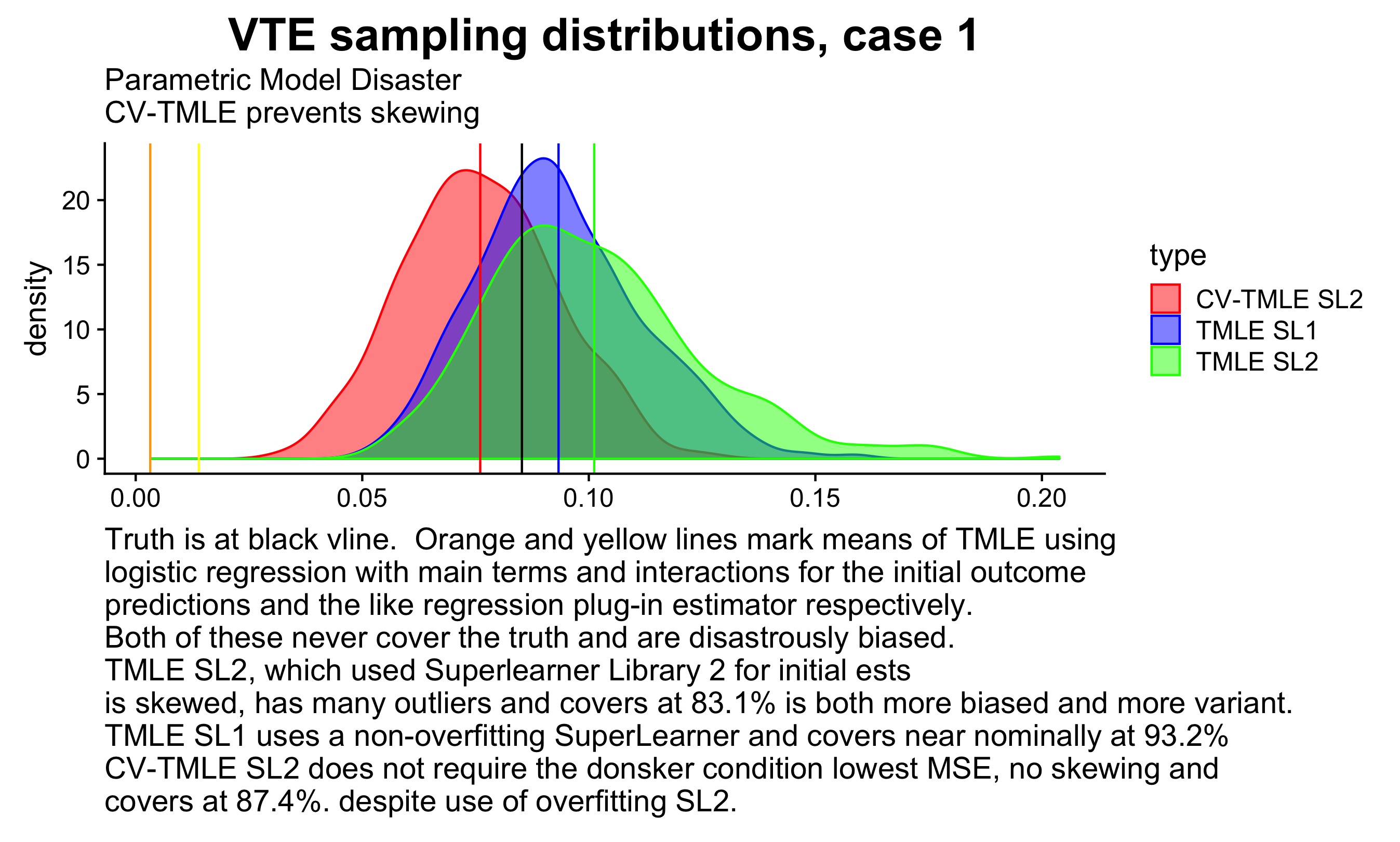}
\end{figure}

\subsection{Mixed Results and Need for Future Refinements}

We again demonstrate how employing targeted learning with CV-TMLE can recover misspecified p-score as well as outcome models when parametric models are terrible, but coverage is below nominal and at times very poor, depending on the situation.  This is not a problem solely for the case of an observational study as the authors have found the main culprit in poor coverage to be the second order remainder term consisting of the integral of true blip function minus estimated blip function squared (see Appendix B for remainder term derviation).  Standard parametric models are again disastrous for both cases 2 and 3, never covering the truth and missing almost all of the VTE as in case 1. Only in case 2 does CV-TMLE, using a pretty small superlearner library of 8 algorithms including, xgboost, neural networks, glm with main terms and interactions, earth, sample mean and the highly adaptive lasso  \parencite{Laan:2015aa}, achieves decent coverage of 83\% and reduces bias of the initial estimate from -0.015 to -0.009.  In case 3, CV-TMLE with the same SuperLearner library only covers at 32\%.\\   

For both case 2 and case 3 we used the following model for the true treatment mechanism:

$\mathbb{E}_0[A \mid W] = expit(.4*(-0.4 * W1*W2 + 0.63 * W2^2 -.66*cos(W1) - 0.25))$

Case 2 true outcome model:

\begin{scriptsize}
$\mathbb{E}_0[Y \mid A,W] = expit(0.1 * W1*W2 + 1.5*A*cos(W1) + 0.15*W1 - .4*W2*(abs(W2) > 1) -1*W2*(abs(W2 <=1)))$
\end{scriptsize}

Case 3 true outcome model:

$\mathbb{E}_0[Y \mid A,W] = expit(0.2 * W1*W2 + 0.1 * W2^2 - .8*A*(cos(W1) + .5*A*W1*W2^2) - 0.35)$

\section{Demonstration on Real Data}

The estimator described in the previous sections was applied to a real dataset. GER-INF is a placebo-controlled randomized trial published in 2002 that evaluated the impact on mortality of low dose steroid administration in patients hospitalized in the intensive care unit (ICU) for septic shock \parencite{gerinf}. This study was performed in 19 ICUs in France and enrolled a total of 299 patients.  Because steroid supplementation in this context was expected to be beneficial in patients with relative adrenal insufficiency, a corticotropin stimulation test was performed in all patients. A pre-specified subgroup analysis was planned in patients who were nonresponders to the corticotropin stimulation test (relative adrenal insufficiency) and in those responders to the corticotropin stimulation test (no relative adrenal insufficiency). In this sample, 7 days of low dose hydrocortisone associated with fludrocortisone were associated with a  reduced risk of death in patients with septic shock. As expected, this reduction was significant in patients with relative adrenal insufficiency as reflected by a lack of response following the corticotropin stimulation test. Because of the apparent heterogeneity in treatment effect, at least partially explained by the presence of a relative adrenal insufficiency, we used the proposed estimator to quantify treatment effect variability across patients. \\

We controlled for the following confounders in fitting the outcome model: SAPS2 severity score (Simplified Acute Physiology Score to assess mortality), Sequential Organ Failure Assessment (SOFA) severity score at baseline
, lactate level,  cortisol level before corticotropin, an indicator of responding to the corticotropin stimulation test, site of infection, mechanical ventilation at baseline, patient origin (hospital acquired infection or not), indicator of elective surgery or urgent surgery, maximum difference in cortisol concentration before and after stimulations, indicator of use of etomidate for anesthesia (drug known to alter the adrenal function), blood sugar and the pathogen responsible for the infection.  In this case, the treatment assignment was random and thus we can identify VTE from the data as a measure of how much of the heterogeneity in treatment effect is due to confounders normally used to assign treatment.  We note, in the case of variables missing data, we create an indicator of missingness and use the median or, in the case of categorical variables, the most popular category as an imputed value.  The table below summarizes our data.\\  

\FloatBarrier
\begin{table}[H] \centering
  \caption*{\begin{scriptsize}\textbf{GER-INF Summary, N=299}\end{scriptsize}}
 {\setstretch{1.0}
\begin{scriptsize}
\begin{tabular}{@{\extracolsep{5pt}} lccc} 
\textbf{Variable} & \multicolumn{1}{c}{} & \multicolumn{1}{c}{}\\ 
\hline \\[-1.8ex] 
\textbf{Outcome: Renal Failure} & \hspace{1 in} & $$ \\ 
\hspace{.2 in}Yes &\hspace{1 in} & $173$ \\ 
\hspace{.2 in}No &\hspace{1 in} & $126$ \\ 
\textbf{Treatment: Rec. Steroid} & \hspace{1 in} & $$ \\ 
\hspace{.2 in}Yes &\hspace{1 in} & $150$ \\ 
\hspace{.2 in}No &\hspace{1 in} & $149$ \\ 
\textbf{Age in yrs} &\hspace{1 in} &60.8 (16.1)  \\
\textbf{IGS2} &\hspace{1 in} & 62.6 (23.6)  \\
\textbf{SOFA0} &\hspace{1 in} & 11 (3.2)  \\
\hspace{.2 in} missing &\hspace{1 in} & $24$ \\ 
\textbf{CORT0} &\hspace{1 in} & 23.3 (30.9)  \\
\textbf{DELTA\_CORTmax} &\hspace{1 in} & 6.1 (22.3)  \\
\textbf{RESPONDER} & \hspace{1 in} & $$ \\ 
\textbf{GLYC} &\hspace{1 in} & 175.8 (106.2)  \\
\hspace{.2 in} missing &\hspace{1 in} & $6$ \\ 
\hspace{.2 in}Yes &\hspace{1 in} & $70$ \\ 
\hspace{.2 in}No &\hspace{1 in} & $229$ \\ 
\textbf{VMO} & \hspace{1 in} & $$ \\ 
\hspace{.2 in}Yes &\hspace{1 in} & $298$ \\ 
\hspace{.2 in}No &\hspace{1 in} & $1$ \\ 
\textbf{origine} & \hspace{1 in} & $$ \\ 
\hspace{.2 in}Yes &\hspace{1 in} & $112$ \\ 
\hspace{.2 in}No &\hspace{1 in} & $187$ \\ 
\textbf{etomidate} & \hspace{1 in} & $$ \\ 
\hspace{.2 in}Yes &\hspace{1 in} & $76$ \\ 
\hspace{.2 in}No &\hspace{1 in} & $213$ \\ 
\textbf{site}&\hspace{1 in} & $$ \\ 
\hspace{.2 in} Multiple &\hspace{1 in} & $144$ \\ 
\hspace{.2 in} Lung &\hspace{1 in} & $89$ \\ 
\hspace{.2 in} GI &\hspace{1 in} & $31$ \\ 
\hspace{.2 in} Soft Tissue &\hspace{1 in} & $16$ \\ 
\hspace{.2 in} Bacteremia &\hspace{1 in} & $6$ \\ 
\hspace{.2 in} Other &\hspace{1 in} & $12$ \\ 
\hspace{.2 in} missing info &\hspace{1 in} & $1$ \\ 
\textbf{typeadmission}&\hspace{1 in} & $$ \\ 
\hspace{.2 in}1: &\hspace{1 in} & $179$ \\ 
\hspace{.2 in}2: &\hspace{1 in} & $10$ \\ 
\hspace{.2 in}3: &\hspace{1 in} & $110$ \\ 
\end{tabular}
\end{scriptsize}}
\end{table}

We provide estimates below for ATE and VTE simultaneously and used the delta method to also give a confidence interval for the $\sqrt{VTE}$, because such is on the scale of measurement of ATE.  We can see the left bounds of the confidence interval for VTE and $\sqrt{VTE}$ strayed into the negative numbers, which are not possible estimates for such parameters.  Log-scaling the confidence intervals only make them excessively large and therefore not useful due to the unscaled confidence bands centering close to 0.  If we reference our discussion about skewing in section 3.2.1, we see that if the true VTE were as small as our estimate, then our sampling distribution under the best case scenario of well-specified outcome model is skewed.  We would need a true VTE of around 0.06 to have any hope, even under the best case scenario of well-specified outcome model, to have normal sampling distributions for estimating VTE.  There is also the issue of second order remainder term bias as we discussed, which could account for missing a truly larger VTE in our estimates.  The second order remainder term in section 2.2.2 is the square of the $L^2$ norm bias in estimating the blip function so, prioritizing the blip function in our outcome model estimation is the subject for future work to improve blip function estimation in finite samples.\\

\begin{table}[!htbp] \centering 
  \caption{CV-TMLE Results for Simultaneous Estimation of ATE, VTE and sqrt(VTE)} 
  \label{} 
\begin{tabular}{@{\extracolsep{5pt}} ccccc} 
\\[-1.8ex]\hline 
\hline \\[-1.8ex] 
 & est & se & lower & upper \\ 
\hline \\[-1.8ex] 
ATE & $$-$0.088$ & $0.050$ & $$-$0.199$ & $0.023$ \\ 
VTE & $0.002$ & $0.004$ & $$-$0.008$ & $0.012$ \\ 
sqrt(VTE) & $0.045$ & $0.048$ & $$-$0.061$ & $0.152$ \\ 
\hline \\[-1.8ex] 
\end{tabular} 
\end{table} 

\subsection{SuperLearner}
We used a SuperLearner library consisting of 40 algorithms, including main terms and interactions when applying any regression methods, such as logistic regression, bayes generalized linear models, lasso and ridge regressions (combinations of $L^1$ and $L^2$ penalty) and earth, which uses data adaptive regression splines.  For boosting trees (xgboost), we used depth 2 trees to allow interaction as well as depth 1 trees with main terms and interactions as the covariates.  Also for boosting, we used different hyperparameters for the number of trees in combination with different learning rates.  We used recursive partitioning and random forest, which are tree methods and therefore account for interactions, as well as neural networks which are more non-parametric approaches.  We also applied k nearest neighbors for prediction.  In addition we applied screening of the top 25 correlated variables with the outcome in conjunction with then running the machine learning algorithm and did the same for the top 5 variables when running bayes generalized linear models as well as glm.  The convex combination of learners, or SuperLearner, had the lowest cross-validated risk (negative log-likelihood) of 0.536 with random forest algorithms and the lasso with all interactions included as variables performing equally well.  The discrete SuperLearner \parencite{Eric-Polley:2017aa}, or the one that chooses the lowest risk algorithm to use for each fold's validation set predictions, had a cross-validated risk of 0.57.  Without knowing which algorithm would perform the best a priori, we can see SuperLearner did the job it was supposed to do in combining our ensemble in an optimal way, according to cross-validated risk.

\section{Discussion}
We can see there are two great challenges in estimating VTE, one being the fact the parameter is bounded below at 0, skewing and biasing the estimates when the true variance is too small for the sample size. In the future we might develop an improvement over log-scaling to form adjusted confidence bounds if they stray into the negative zone.  To our eyes, this problem is not as crucial as obtaining reliable inference when the true variance is large enough to be estimated for the sample size. On this front the second order remainder term, $-\mathbb{E}_{0}\left(b_{0}(W)-b(W)\right)^{2}$, has proven difficult to contain in finite samples.  We are certain for larger samples we can show that a Superlearner library including the highly adaptive lasso \parencite{Benkeser:2016aa} will achieve the necessary rates of convergence for this term to be truly second order, however, such is not trivial to show conclusively via very time-consuming and computer intensive simulations and besides, a sample size of 1000 for 4 covariates and treatment is certainly of practical importance.  To this end, the authors plan to propose in a follow-up paper, a novel way to estimate the blip function that will also yield estimates of the outcome predictions that are necessarily between 0 and 1, all the while performing asymptotically as good as just fitting the outcome model. We need blip function estimates, $b(W)$, that are compatible with the outcome predictions in order to perform the crucial targeting step of the TMLE procedure. i.e., we need $b(W) = \bar{Q}(1, W) - \bar{Q}(0,W)$.\\

Another approach to yielding better coverage would be to account for the second order remainder term via the use of the new highly adaptive lasso (HAL) non-parametric bootstrap  \parencite{2017arXiv170809502V}, to form our confidence intervals.  HAL, as mentioned previously in this paper, guarantees the necessary rates of convergence of the second order remainder terms under very weak conditions and is thus guaranteed to yield asymptotically efficient TMLE estimates \parencite{Laan:2015ab}, using the empirical variance of the efficient influence curve approximation for the standard error.  However, in finite samples such a procedure yields lower than nominal coverage due to the unforgiving second order remainder term of non-doubly robust estimators as we have for VTE.  We feel the HAL CV-TMLE, using a non-parametric bootstrap addresses this issue and accounts for the second order bias, the subject of more future work.  Perhaps best will be performing the novel blip fitting procedure with the HAL bootstrap.   

\clearpage
\appendix{\begin{Huge}\textbf{Appendix}\end{Huge}}
\section{The Influence Curve and Remainder Term Derivations}
\subsubsection*{\noindent Set Up}
Our observed data, O, is of the form, O = (W,A,Y), where W is a set of confounders, A is a binary treatment indicator and Y is the outcome, continuous or binary.  \noindent $O\sim P\in\mathcal{M},\text{ non-parametric}$.  It is important that we can factorize the density for $P$ is as $p(o)=p_{Y}(y\vert a,w)g(a\vert w)p_{W}(w)$.

\subsubsection{Tangent Space for Nonparametric Model}
We consider the one dimensional set of submodels that pass through $P$ at $\epsilon=0$ \parencite{Vaart:2000aa} $\{P_\epsilon \text{ def. by density, }p_\epsilon=(1+\epsilon S)p\vert \int SdP=0,\int S^2dP<\infty \}$.
The tangent space is the closure in $L^2$ norm of the set of scores, $S$, or directions for the paths defined above. We write:
\begin{eqnarray*}
T & = & \overline{\left\{ S(o)\vert\mathbb{E}S=0,\mathbb{E}S^2<\infty\right\}} \\
 & = & \overline{\{S(y\vert a,w)\vert\mathbb{E}_{P_{Y}}S=0,\mathbb{E}S^2<\infty\}}\oplus\overline{\{S(a\vert w)\vert\mathbb{E}_{P_{A}}S=0,\mathbb{E}S^2<\infty\}}\oplus\overline{\{S(w)\vert\mathbb{E}_{P_{W}}S=0,\mathbb{E}S^2<\infty\}}\\
  & = & T_{Y}\oplus T_{A}\oplus T_{W}
\end{eqnarray*}

For a non-parametric model, $T = L^{2}_{0}(P)$ forms a Hilbert space with inner product defined as $\langle f, g\rangle=\mathbb{E}_Pfg$.  Our notion of orthogonality now is $f \perp g$ if and only if $\langle f, g \rangle = 0$ and, therefore, the above direct sum is valid.  In other words, every score, $S$, can be written as $ \frac{d}{d\epsilon}log(p_{\epsilon})\vert_{\epsilon=0}=S(w,a,y)=S_Y(y\vert a,w)+S_A(a\vert w) + S_W(w)$ where, due to the fact $p_\epsilon=(1+\epsilon S)p=p_{Y\epsilon}p_{A\epsilon}p_{W\epsilon}$, it is easy to see $\frac{d}{d\epsilon}log(p_{Y\epsilon})\vert_{\epsilon=0}=S_Y(y\vert a,w)$, $\frac{d}{d\epsilon}log(p_{A\epsilon})\vert_{\epsilon=0}=S_A(a\vert w)$ and $\frac{d}{d\epsilon}log(p_{W\epsilon})\vert_{\epsilon=0}=S_W(w)$.  Furthermore we know that a projection of $S$ on $T_Y$ is  

\begin{eqnarray*}
S_Y(y \mid w, a) &= &S(w,a,y) - E[S(W, A, Y) \mid W = w, A = a] \\
& = & S(w,a,y) - \int S(w,a,y) p_Y(y \mid w, a) d\nu(y)\\
& = & \frac{d}{d\epsilon}log(p_{Y\epsilon})\vert_{\epsilon=0}
\end{eqnarray*}

\subsubsection{Efficiency Theory in brief}
Our parameter of interest is a mapping from the model, $\mathcal{M}$ to the real numbers given by $\Psi(P)=\mathbb{E}(b(W)-\mathbb{E}b)^2$ where $b(W)=\mathbb{E}[Y\vert A=1,W]-\mathbb{E}[Y\vert A=0,W]$.  We can borrow from van der Vaart, 2000, who defines an influence function, otherwise known as a gradient, as a continuous linear map from $T$ to the reals given by
\begin{equation}
\underset{\epsilon\rightarrow 0}{lim}\left(\frac{\Psi(P_\epsilon)-\Psi(P)}{\epsilon}\right)\longrightarrow \dot{\Psi}_{P}(s)
\end{equation}
We note to the reader, we imply a direction, $S$, when we write $P_{e}$, which has density $p(1+\epsilon S)$, but generally leave it off the notation as understood.\\

By the riesz representation theorem \parencite{riesz} for Hilbert Spaces, assuming the mapping in (1) is a bounded and linear functional on $T$, it can be written in the form of an inner product $\langle D^*(P),g \rangle$ where $D^*$ is a unique element of $T$, which we call the canonical gradient or efficient influence curve.  Thus, in the case of a nonparametric model, the only gradient is the canonical gradient.  It is notable that the efficient influence curve has a variance that is the lower bound for any regular asymptotically linear estimator \parencite{Vaart:2000aa}. Since the TMLE, under conditions as discussed in this paper, asymptotically achieves variance equal to that of the efficient influence curve, the estimator is asymptotically efficient.\\

As a note to the reader: Our parameter mapping does not depend on the treatment mechanism, $g$, and also $T_{A}\perp T_{Y}\oplus T_{W}$ which, means our efficient influence curve must therefore be in $T_{Y}\oplus T_{W}$ for the nonparametric model.  Therefore, our efficient influence curve will have two orthogonal components in $T_Y$ and $T_W$ respectively. We have no component in $T_A$, which is why we need not perform a TMLE update of the initial prediction, $g_n$, of $g_0(A\vert W)$. Such also teaches us that for the semi-parametric model, where the treatment mechanism is known, the efficient influence function will be the same as for the non-parametric model.\\  

\subsection{\noindent Derivation of the efficient influence curve}
For the below proof, when we take the derivatives we will assume that such is in the direction of a given score, $S\in T$.  We will assume a dominating measure $\nu$ and merely use $\nu$ to always denote the dominating measure of all densities involved. We could perhaps just assume continuous densities and use lebesque measure as well and nothing would change below. \\

\begin{thm}
Let $\Psi(P)=var_{P}(b(W))$.
\noindent The efficient influence curve for $\Psi$ at $P$ is given
by:
\[
\mathbf{D^{\star}(P)(W,A,Y)=}\mathbf{\mathbf{2\left(b(W)-\mathbb{E}b(W)\right)\left(\frac{2A-1}{g(A\vert W)}\right)\left(Y-\bar{Q}(A,W)\right)+\left(\mathbf{b(W)}-\mathbb{E}b\right)^{2}-\varPsi(P)}}
\]
where $\bar{Q}(A,W)=\mathbb{E}(Y\vert A,W)$
\end{thm}

\begin{proof}
\noindent by our previous discussion, we are guaranteed a
unique representer in $L_{0}^{2}(P)$, called $D^{\star}$ such that
\[
\frac{d}{d\epsilon}\Psi(P_{\epsilon})(S)\biggr\vert_{\epsilon=0}=\langle D^{\star},S\rangle=\mathbb{E}D^{\star}S
\]

\noindent We will now write $\frac{d}{d\epsilon}\Psi(P_{\epsilon})(S)\biggr\vert_{\epsilon=0}\text{ as }\langle D^{\star},S\rangle$
and $D^{\star}$ will be our desired formula:

\noindent
\begin{eqnarray}
\frac{d}{d\epsilon}\Psi(P_{\epsilon})(S)\biggr\vert_{\epsilon=0} & = & \frac{d}{d\epsilon}\mathbb{E}_{P_{\epsilon}}\left(b_{P_{\epsilon}}(W)-\mathbb{E}_{P_{\epsilon}}b_{P_{\epsilon}}(W)\right)^{2}\biggr\vert_{\epsilon=0}\nonumber \\
 & = & \frac{d}{d\epsilon}\int\left(b_{P_{\epsilon}}(w)-\mathbb{E}_{P_{\epsilon}}b_{P_{\epsilon}}(W)\right)^{2}p_{\epsilon}(o)\nu(do)\biggr\vert_{\epsilon=0}\nonumber \\
 & = & \int2\left(b_{P_{\epsilon}}(w)-\mathbb{E}_{P_{\epsilon}}b_{P_{\epsilon}}(W)\right)\frac{d}{d\epsilon}\left(b_{P_{\epsilon}}(w)-\mathbb{E}_{P_{\epsilon}}b_{P_{\epsilon}}\right)p(o)\nu(do)\biggr\vert_{\epsilon=0}+\mathbb{E}\left[\left(b(W)-\mathbb{E}b(W)\right)^{2}S(O)\right]\nonumber \\
 & = & \int2\left(b_{P}(w)-\mathbb{E}_{P}b_{P}(W)\right)\frac{d}{d\epsilon}b_{P_{\epsilon}}(w)p(o)\nu(do)\biggr\vert_{\epsilon=0}+\mathbb{E}\left[\left(\left(b(W)-\mathbb{E}b(W)\right)^{2}-\Psi(P)\right)S(O)\right]
\end{eqnarray}

\noindent Now we can compute the first term in (4)

\begin{eqnarray}
 &  & 2\int\left(b_{P_{\epsilon}}(w)-\mathbb{E}_{P_{\epsilon}}b_{P_{\epsilon}}(W)\right)\frac{d}{d\epsilon}\left[\int\left(yp_{Y\epsilon}(y\vert a=1,w)-yp_{Y\epsilon}(y\vert A=0,w)\right)\nu(dy)\right]p_{W}(w)\nu(dw)\biggr\vert_{\epsilon=0}\nonumber \\
 & = & 2\int\left(b_{P_{\epsilon}}(w)-\mathbb{E}_{P_{\epsilon}}b_{P_{\epsilon}}(W)\right)\frac{d}{d\epsilon}\left[\int\frac{2a-1}{g(a\vert w)}p_{Y\epsilon}(y\vert a,w)g(a\vert w)\nu(d(y\times a))\right]p_{W}(w)\nu(dw)\biggr\vert_{\epsilon=0}\nonumber \\
 &  & 2\int\left(b_{P_{\epsilon}}(w)-\mathbb{E}_{P_{\epsilon}}b_{P_{\epsilon}}(W)\right)\left[\int\frac{2a-1}{g(a\vert w)}\frac{d}{d\epsilon}\frac{p_{\epsilon}(w,a,y)}{p_{A\epsilon}(a\vert w)p_{W\epsilon}(w)}\biggr\vert_{\epsilon=0}g(a\vert w)\nu(d(y\times a))\right]p_{W}(w)\nu(dw)\\
 & = & 2\int\left(b_{P_{\epsilon}}(w)-\mathbb{E}_{P_{\epsilon}}b_{P_{\epsilon}}(W)\right)\frac{y(2a-1)}{g(a\vert w)}\left(S(w,a,y) - \int S(w,a,y) p_Y(y \mid w, a) d\nu(y)\right)p(w,a,y)\nu(do))\nonumber \\
 & \overset{fubini}{=} & 2\int\left(b_{P_{\epsilon}}(w)-\mathbb{E}_{P_{\epsilon}}b_{P_{\epsilon}}(W)\right)\int\frac{(2a-1)}{g(a\vert w)}(y-\bar{Q}(a,w))S(w,a,y)p(w,a,y)\nu(d(o))\nonumber \\
 & = & 2\int\left(b_{P_{\epsilon}}(w)-\mathbb{E}_{P_{\epsilon}}b_{P_{\epsilon}}(W)\right)\int\frac{(2a-1)}{g(a\vert w)}(y-\bar{Q}(a,w))S(o)p(o)\nu(d(o))
\end{eqnarray}

\[
2\left(b(W)-\mathbb{E}b(W)\right)\left(\frac{2A-1}{g(A\vert W)}\right)(Y-\bar{Q}(A,W))+\left(b(W)-\mathbb{E}b\right)^{2}-\Psi(P)
\]
is the aforementioned representer, completing the proof.
\end{proof}

\subsection{\noindent Remainder terms and asymptotic linearity:}
The reader may recall the three conditions assuring asymptotic efficiency of the TMLE estimator.  Here we will focus on the $2^{nd}$ condition regarding the remainder term.

\begin{thm}
If $P_0$ is the true distribution, it is necessary to estimate the true blip function $b_0$ at a rate of $\frac{1}{n^{0.25}}$ in the $L^{2}(P)$ norm in order for TMLE to be a consistent asymptotically efficient estimator under a known treatment mechanism, $g_0$. If $g_0$ is unknown, we also need the product of the $L^{2}(P_0)$ rates for estimating $g_0$ and $\bar{Q}_0$ to be $\frac{1}{n^{0.5}}$.
\end{thm}

\begin{proof}

\noindent For this discussion we will drop the subscript, n, and superscript, $\star$ in $P_{n}^{\star}$
and merely consider, $P$, as an estimate of the truth, $P_0$. We will use $b(W)$ to denote the TE function where the conditional expectation is with respect to distribution, $P$, ie the estimated TE function, and $b_0(W)$ to be the true TE function. Likewise, $\mathbb{E}_0$ is the expectation with respect to the true observed data distribution, $P_0$, and leaving the subscript, $0$, off the expectation sign means the expectation is with respect to $P$.  

\noindent {\scriptsize{}
\begin{eqnarray}
R_{2}(P,P_{0}) & = & \Psi(P)-\Psi(P_{0})+P_{0}\left(D^{\star}(P)\right)\nonumber \\
 & = & \mathbb{E}\left(b(W)-\mathbb{E}b(W)\right)^{2}-\mathbb{E}_{0}\left(b_{0}(W)-\mathbb{E}_{0}b_{0}(W)\right)^{2}+\nonumber \\
 &  & \mathbb{E}_{0}\left[2\left(b(W)-\mathbb{E}b(W)\right)\frac{2A-1}{g(A\vert W)}\left(Y-\bar{Q}(A,W\right)+\left(b(W)-\mathbb{E}b(W)\right)^{2}-\Psi(P)\right]\nonumber \\
 & = & -\mathbb{E}_{0}\left(b_{0}(W)-\mathbb{E}_{0}b_{0}(W)\right)^{2}+\nonumber \\
 &  & \mathbb{E}_{0}\left[2\left(b(W)-\mathbb{E}b(W)\right)\frac{2A-1}{g(A\vert W)}\left(Y-\bar{Q}(A,W)\right)+\left(b(W)-\mathbb{E}b(W)\right)^{2}\right]\nonumber \\
 & = & \mathbb{E}_{0}\left[\left(b(W)-\mathbb{E}b(W)\right)^{2}-\left(b_{0}(W)-\mathbb{E}_{0}b_{0}(W)\right)^{2}\right]+\nonumber \\
 &  & \mathbb{E}_{0}\mathbb{E}_{0}\left[2\left(b(W)-\mathbb{E}b(W)\right)\frac{2A-1}{g(A\vert W)}\left(\bar{Q}_{0}(A,W)-\bar{Q}(A,W)\right)\vert W\right]\nonumber \\
 & = & \mathbb{E}_{0}\left[\left(b(W)-\mathbb{E}b(W)\right)^{2}-\left(b_{0}(W)-\mathbb{E}_{0}b_{0}(W)\right)^{2}\right]+\nonumber \\
 &  & +\mathbb{E}_{P_{W}}\left[2\left(b(W)-\mathbb{E}b(W)\right)\left(\frac{g_{0}(1\vert W)}{g(1\vert W)}\left(\bar{Q}_{0}(1,W)-\bar{Q}(1,W)\right)-\frac{g_{0}(0\vert W)}{g(0\vert W)}\left(\bar{Q}_{0}(0,W)-\bar{Q}(0,W)\right)\right)\right]\nonumber \\
 & = & \mathbb{E}_{0}\left[\left(b(W)-\mathbb{E}b(W)\right)^{2}-\left(b_{0}(W)-\mathbb{E}_{0}b_{0}(W)\right)^{2}+2\left(b_{0}(W)-b(W)\right)\left(b(W)-\mathbb{E}b(W)\right)\right]\nonumber \\
 &  & +\mathbb{E}_{0}\left[2\left(b(W)-\mathbb{E}b(W)\right)\left(\frac{g_{0}(1\vert W)-g(1\vert W)}{g(1\vert W)}\left(\bar{Q}_{0}(1,W)-\bar{Q}(1,W)\right)-\frac{g_{0}(0\vert W)-g(0\vert W)}{g(0\vert W)}\left(\bar{Q}_{0}(0,W)-\bar{Q}(0,W)\right)\right)\right]\nonumber\\
 & = & \left(\mathbb{E}_{0}b_{0}(W)-\mathbb{E}b(W)\right)^{2}-\mathbb{E}_{0}\left(b_{0}(W)-b(W)\right)^{2}\\
 &  & +\mathbb{E}_{0}\left[2\left(b(W)-\mathbb{E}b(W)\right)\left(\frac{g_{0}(1\vert W)-g(1\vert W)}{g(1\vert W)}\left(\bar{Q}_{0}(1,W)-\bar{Q}(1,W)\right)-\frac{g_{0}(0\vert W)-g(0\vert W)}{g(0\vert W)}\left(\bar{Q}_{0}(0,W)-\bar{Q}(0,W)\right)\right)\right]\nonumber
\end{eqnarray}
}{\scriptsize \par}

We can regard the $\left(\mathbb{E}_{0}b_{0}(W)-\mathbb{E}b(W)\right)^{2}$
term in (6) and notice that for an unknown, $g_{0}$, it is well-known that
the double robustness of TMLE in estimating the causal risk difference,
$\mathbb{E}_{0}b_{0}(W)$, implies that if we estimate both $g_{0}$
and $\bar{Q}_{0}$ so that the product of the respective $L_{2}$
rates of convergence is $o(n^{-0.5})$, then we obtain $\sqrt{n}\left(\mathbb{E}_{0}b_{0}(W)-\mathbb{E}b(W)\right)\overset{D}{\implies}N\left[0,var_{0}(D_{1}^{\star}(P_{0}))\right]$
where $D_{1}^{\star}(P_{0})$ is the efficient influence curve for
the causal risk difference. We therefore know $\mathbb{E}_{0}b_{0}(W)-\mathbb{E}b(W)\overset{p}{\longrightarrow}0$
and by slutsky's theorem, $\sqrt{n}\left(\mathbb{E}_{0}b_{0}(W)-\mathbb{E}b(W)\right)^{2}\overset{D}{\implies}0$.
Therefore this term poses no additional problem to the rest of the terms. \\

Now we can address the standard "double robust" term:
\begin{footnotesize}
\begin{align*}
 & \mathbb{E}_{0}\left[2\left(b(W)-\mathbb{E}b(W)\right)\left(\frac{g_{0}(1\vert W)-g(1\vert W)}{g(1\vert W)}\left(\bar{Q}_{0}(1,W)-\bar{Q}(1,W)\right)-\frac{g_{0}(0\vert W)-g(0\vert W)}{g(0\vert W)}\left(\bar{Q}_{0}(0,W)-\bar{Q}(0,W)\right)\right)\right]\\
\leq & K\mathbb{E}_{0}\left[\biggr\vert\frac{g_{0}(1\vert W)-g(1\vert W)}{g(1\vert W)}\left(\bar{Q}_{0}(1,W)-\bar{Q}(1,W)\right)\biggr\vert+\biggr\vert\frac{g_{0}(0\vert W)-g(0\vert W)}{g(0\vert W)}\left(\bar{Q}_{0}(0,W)-\bar{Q}(0,W)\right)\biggr\vert\right]\\
\leq & K\mathbb{E}_{0}\biggr\vert\frac{g_{0}(A\vert W)-g(A\vert W)}{g(A\vert W)g_{0}(A\vert W)}\left(\bar{Q}_{0}(A,W)-\bar{Q}(A,W)\right)\biggr\vert\\
\leq & K\Vert g_{0}(A\vert W)-g(A\vert W)\Vert_{L^{2}(P_{0})}\Vert\bar{Q}_{0}(A,W)-\bar{Q}(A,W)\Vert_{L^{2}(P_{0})}
\end{align*}
\end{footnotesize}
where the last inequality follows from cauchy-schwarz and the strict positivity assumption on $g_{0}$. The difficult term in (6) is $\mathbb{E}_{0}\left(b_{0}(W)-b(W)\right)^{2}$ but if we obtain the required $L^2$ rates it is obvious this term will be second order and we have proven the theorem.  
\end{proof}

\section{Logistic Regression Plug-in Estimator Inference}

\subsection{defining the parameter of interest}
Let us consider the nonparametric model, $\mathcal{M}$.  We will employ basic statistics to obtain the influence curve for a logistic regression plug-in estimator for both the mean and variance of the blip function, $b(W) = \mathbb{E}[Y \mid A=1, W] -  \mathbb{E}[Y \mid A=0, W]$. We define the plug-in estimator as plugging in the outcome conditional density given by the MLE for $\beta$ where 
\[
\beta = \underset{\gamma}{argmin}[-\mathbb{E}_{P_{A,W}}\mathbb{E}_{P_{\gamma, Y \mid A,W}}log(p_{\gamma}(Y \mid A,W)p_{A,W}(A,W))]
\] 
$Y$ and $A$ are binary and $P_{\gamma}(Y \mid A,W)$ is defined as a conditional density, 
\[
p_{\gamma}( Y \mid A,W) = expit\left(m(A,W\vert\gamma)\right)^Y(1-expit\left(m(A,W\vert\gamma)\right)^{1-Y}
\]
for a fixed function $m(\cdot \vert \cdot)$.

$\Psi_1(P) = \mathbb{E}_{W}b_\beta(W)$ and 

$\Psi_2(P) = \mathbb{E}_{W}(b_\beta(W) - \Psi_1(P))^2$

and we may consider the two dimensional parameter: $\Psi(P) = (\Psi_1(P), \Psi_2(P))$.

We note to the reader, that $b_\beta(W) = expit\left(m(1,W\vert\beta)\right) - expit\left(m(0,W \vert\beta)\right)$ is the conditional average treatment effect under the parametric model for strata, $W$.  

\subsection{Finding the MLE for the coefficients, $\beta$}
Now if we take n iid draws of $O$ we get that the likelihood of drawing
$\{O_{i}\}_{i=1}^{n}$ is 
\[
\prod_{i=1}^{n}expit\left(m(A_{i},W_{i}\vert\beta\right)^{Y_{i}}\left(1-expit\left(m(A_{i},W_{i}\vert\beta\right)\right)^{1-Y_{i}}p_{A}(A_{i}\vert W_{i})p_{W}(W_{i})
\]

We thus have:
\begin{eqnarray}
\bigtriangledown_{\beta}log\prod_{i=1}^{n}expit\left(m(A_{i},W_{i}\vert\beta\right)^{Y_{i}}\left(1-expit\left(m(A_{i},W_{i}\vert\beta\right)\right)^{1-Y_{i}}p_{A}(A_{i}\vert W_{i})p_{W}(W_{i}) & =\nonumber \\
\nabla_{\beta}\sum_{i=1}^{n}\left[Y_{i}log\left(expit\left(m(A_{i},W_{i}\vert\beta\right)\right)+(1-Y_{i})log\left(1-expit\left(m(A_{i},W_{i}\vert\beta\right)\right)\right] & =\\
\sum_{i=1}^{n}\left(\frac{\partial m}{\partial\beta_{0,n}},\frac{\partial m}{\partial\beta_{1,n}},...,\frac{\partial m}{\partial\beta_{d,n}}\right)^{T}\left(Y_{i}-expit\left(m(A_{i},W_{i}\vert\beta_{n}\right)\right) & = 0
\end{eqnarray}

We can now derive the multidimensional influence function via the use of a taylor series about $\mathbb{E}_{P_{\beta}}S_{\beta}(O)$
where 
\[
S_{\beta}(O)=\left(\frac{\partial m}{\partial\beta_{0}},\frac{\partial m}{\partial\beta_{1}},...,\frac{\partial m}{\partial\beta_{d}}\right)^{T}\left(Y-expit\left(m(A,W\vert\beta\right)\right)
\]

Also note that the derivative of the log-likelihood or score, $S_{\beta}(O)$, has mean 0 by assumption.
So $P_{\beta}S_{\beta}(O)=\mathbb{E}_{P_{\beta}}S_{\beta}(O)=0$. 

Thus we get by virtue of $\beta_n$ being the MLE: {\scriptsize{}
\begin{eqnarray}
P_{n}S_{\beta_{n}}(O)-P_{\beta}S_{\beta}(O) & = & 0\nonumber \\
P_{n}S_{\beta_{n}}(O)-P_{\beta}S_{\beta}(O)+P_{\beta}S_{\beta_{n}}(O)-P_{\beta}S_{\beta_{n}}(O) & = & 0\nonumber \\
(P_{n}-P_{\beta})S_{\beta_{n}}(O) & = & P_{\beta}\left(S_{\beta}(O)-S_{\beta_{n}}(O)\right)\nonumber \\
\sqrt{n}(P_{n}-P)S_{\beta_{n}}(O) & = & -\sqrt{n}P_{\beta}\left(\nabla_{\beta}S_{\beta}(O)\right)(\beta_{n}-\beta)+\sqrt{n}O_{p}\Vert\beta_{n}-\beta\Vert^{2}\\
\implies\sqrt{n}(\beta_{n}-\beta) & \overset{D}{\implies} & \sqrt{n}(P_{n}-P_{\beta})\left(-P_{\beta}\left(\nabla_{\beta}S_{\beta}(O)\right)^{-1}S_{\beta}(O)\right)
\end{eqnarray}
}{\scriptsize \par}

since we can consider the $\Vert \beta_n-\beta \Vert ^2$ term as second order.  Therefore 

\[
IC_{\beta_n}(O) = \left(-P_{\beta}\left(\nabla_{\beta}S_{\beta}(O)\right)^{-1}S_{\beta}(O)\right)
\]

is the influence curve for the maximum likelihood estimator of the truth, $\beta$. In the case of logistic regression we have $m(A,W\vert\beta)=X(A,W)^T\beta$ where we might have the main terms linear case, $X(A,W)^T=(1,A,W)^{T}$ and we consider $\beta$ as a column vector of coefficients, including the intercept, $\beta_{0}.$  However, $X(A,W)^T$ might be any combination of columns of the covariates, as in any of variables, containing interactions and so forth of the main terms in the right-hand side of our regression formula. For now we will drop the arguments in X(A,W) and just use X unless it is necessary. 

\[
S_{\beta}(O)=X\left(Y-expit(\beta^T X)\right)
\]

\[
\nabla_{\beta}S_{\beta}(O)=\left[\begin{array}{c}
\nabla_{\beta}^{T}S_{\beta,0}\\
.\\
\nabla_{\beta}^{T}S_{\beta,d}
\end{array}\right]=expit(\beta^T X)(1-expit(\beta^T X))XX^{T}
\]

\subsection{Influence curve for MLE esimate of $b_{\beta}(W_0)$}
Consider the parameter $\Psi_{a,w}(P)=expit\left(m(a,w\vert\beta)\right)$ for fixed $(a, w)$. $expit\left(m(a,w\vert\beta\right)$ is a continuously differentiable function of $\beta,$ which means we can apply the ordinary delta method as follows to find the plug-in estimator influence curve estimating $\Psi_{a,w}(P)$. 

\begin{align*}
\bar{Q}_{\beta_{n}}(a,w)-\bar{Q}_{\beta}(a,w) & =\nabla_{\beta}^{T}\bar{Q}_{\beta}(a,w)\sum_{i=1}^{n}IC_{\beta_{n}}(O_i)+R_{2}(P_{\beta},P_{\beta_{n}})\\
 & = expit(\beta^{T}x)(1-expit(\beta^{T}x))x^{T}\sum_{i=1}^{n}IC_{\beta_n}(O_i)+R_{2}(P_{\beta},P_{\beta_{n}})
\end{align*}

where $x = X(a,w)$ and $R_{2}=o_{p}(n^{-.5})$.\\

We thus have the influence curve for the logistic regression MLE plug-in estimator for $b_\beta(W_0)(O)$, notated as $IC_{b_{\beta_n}(W_0)}(O)$:

\begin{scriptsize}
\begin{eqnarray*}
IC_{b_{\beta_n}(W_0)}(O) & =\\
\left[expit\left(\beta^{T}X(1,W_{0})\right)\left(1-expit\left(\beta^{T}X(1,W_{0})\right)\right)X(1,W_{0})^{T}-expit\left(\beta^{T}X(0,W_{0})\right)\left(1-expit\left(\beta^{T}X(0,W_{0})\right)\right)X(0,W_{0})^{T}\right]IC_{\beta_{n}}(O) & =\\
f_{\beta}(W_{0})IC_{\beta_{n}}(O)
\end{eqnarray*}
\end{scriptsize}

\subsection{Influence curve for MLE estimate of $b_\beta(W_0)^2$}
by the delta method we easily get

\[
IC_{b_{\beta_n}(W_0)^2}=2b_\beta(W_0)IC_{b_\beta(W_0)}(O)
\]

\subsection{Telescoping to find the 2-d influence curve for 2-d parameter $\Psi(P)$}

Note, that $\Psi_n = (\Psi_{1,n},\Psi_{2,n})$ is the MLE plug-in estimate for $\Psi(P) = (\Psi_{1},\Psi_{2})$.\\

\begin{eqnarray*}
\Psi_{1,n}-\Psi_{1} & = & \frac{1}{n}\sum_{i=1}^{n}\left(b_{\beta_{n}}(W_{i})-\Psi_{1}\right)\\
 & = & \frac{1}{n}\sum_{i=1}^{n}\left(b_{\beta_{n}}(W_{i})-b_{\beta}(W_{i})\right)+\underbrace{\frac{1}{n}\sum_{i=1}^{n}b_{\beta}(W_{i})-\Psi_{1}}_{\text{set aside}}
\end{eqnarray*}

Now we can ignore the terms that are set aside as they are part of
the influence curve in the tangent space of mean 0 functions of $W$. 

\begin{eqnarray*}
\frac{1}{n}\sum_{i=1}^{n}\left(b_{\beta_{n}}(W_{i})-b_{\beta}(W_{i})\right) & =\\
\frac{1}{n}\sum_{i=1}^{n}\left(\frac{1}{n}\sum_{j=1}^{n}IC_{b_{\beta_{n}}(W_{i})}(O_{j})\right)+o_{p}(n^{-0.5}) & =\\
\frac{1}{n}\sum_{i=1}^{n}f_{\beta}(W_{i})\left(\frac{1}{n}\sum_{j=1}^{n}IC_{\beta_{n}}(O_{j})\right)+o_{p}(n^{-0.5}) & =\\
Pf_{\beta}(W)\left(\frac{1}{n}\sum_{j=1}^{n}IC_{\beta_{n}}(O_{j})\right)+\underbrace{\left(P_{n}-P\right)f_{\beta}(W)\left(\frac{1}{n}\sum_{j=1}^{n}IC_{\beta_{n}}(O_{j})\right)}_{o_{p}(n^{-0.5})}+o_{p}(n^{-0.5}) & =\\
\mathbb{E}f_{\beta}(W)\left(\frac{1}{n}\sum_{j=1}^{n}IC_{\beta_{n}}(O_{j})\right)+o_{p}(n^{-0.5})
\end{eqnarray*}

Therefore, the influence for the MLE-based estimate of $\Psi_{1}(P),$denoted
by $IC_{\Psi_{1,n}}$, is 
\[
\mathbf{IC_{\Psi_{1,n}}(O)=\mathbb{E}f_{\beta}(W)IC_{\beta_{n}}(O)}
\]

\begin{eqnarray*}
\Psi_{2,n}-\Psi_{2} & = & \frac{1}{n}\sum_{i=1}^{n}\left(b_{\beta_n}(W_{i})-\Psi_{1,n}\right)^{2}-\Psi_{2}\\
 & = & \frac{1}{n}\sum_{i=1}^{n}\left[\left(b_{\beta_n}(W_{i})-\Psi_{1,n}\right)^{2}-\left(b_\beta(W_{i})-\Psi_{1}\right)^{2}\right]+\underbrace{\frac{1}{n}\sum_{i=1}^{n}\left(b_\beta(W_{i})-\Psi_{1}\right)^{2}-\Psi_{2}}_{\text{set aside}}
\end{eqnarray*}

regarding the terms not set aside:

\begin{eqnarray*}
\frac{1}{n}\sum_{i=1}^{n}\left[\left(b_{\beta_n}(W_{i})-\Psi_{1,n}\right)^{2}-\left(b_\beta(W_{i})-\Psi_{1}\right)^{2}\right] & =\\
\frac{1}{n}\sum_{i=1}^{n}b_{\beta_n}(W_{i})^{2}-b_\beta(W_{i})^{2}-\left(\frac{1}{n}\sum_{i=1}^{n}b_{\beta_n}(W_{i})\right)^{2}+\Psi_{1}^{2} & =\\
\frac{1}{n}\sum_{i=1}^{n}\frac{1}{n}\sum_{j=1}^{n}IC_{b_\beta(W_{i})}(O_j)-\frac{2}{n}\Psi_{1}\sum_{i=1}^{n}IC_{\Psi_{1}}(O_{i})+o_{p}(n^{-0.5}) & =\\
\frac{1}{n}\sum_{i=1}^{n}2b_\beta(W_{i})f_{\beta}(W_{i})\frac{1}{n}\sum_{j=1}^{n}IC_{\beta_{n}}(O_j)-\frac{2}{n}\Psi_{1}\sum_{i=1}^{n}IC_{\Psi_{1}}(O_{i})+o_{p}(n^{-0.5}) & =\\
P\left(b_\beta(W)f_{\beta}(W)\right)\frac{2}{n}\sum_{j=1}^{n}IC_{\beta_{n}}(O_j)-\frac{2}{n}\Psi_{1}\sum_{i=1}^{n}IC_{\Psi_{1}}(O_{i})+o_{p}(n^{-0.5}) + & \\
\underbrace{\left[\frac{1}{n}\sum_{i=1}^{n}2b_\beta(W_{i})f_{\beta}(W_{i})-P\left(b_\beta(W)f_{\beta}(W)\right)\right]}_{o_{p}(n^{-0.5})}\frac{2}{n}\sum_{j=1}^{n}IC_{\beta_{n}}(O_j) & =\\
P\left(b_\beta(W)f_{\beta}(W)\right)\frac{2}{n}\sum_{j=1}^{n}IC_{\beta_{n}}(O_{i})-\frac{2}{n}\Psi_{1}\sum_{i=1}^{n}IC_{\Psi_{1}}(O_{i})+o_{p}(n^{-0.5})
\end{eqnarray*}
\begin{align*}
IC_{\Psi_{2,n}}(O) & =2\mathbb{E}\left[b_{\beta}(W)f_{\beta}(W)\right]IC_{\beta_{n}}(O)-2\Psi_{1}IC_{\Psi_{1}}(O)\\
 & =2\mathbb{E}\left[\left(b_{\beta}(W)-\Psi_{1}\right)f_{\beta}(W)\right]IC_{\beta_{n}}(O)
\end{align*}

Thus we arrive at the following influence curve for the plug-in estimator of the
two dimensional parameter $(\Psi_1(P), \Psi_2(P))$.  Note, this IC below is written as a sum of two 2-dimensional vectors, one for the components of the IC in $T_Y$ and the other for the components in $T_W$.  
\[ \mathbf{IC_{\Psi_n}(O) =
\mathbb{E}\left(f_{\beta}(W),2(b_\beta(W)-\Psi_1(P))f_{\beta}(W)\right)IC_{\beta_n}(O)+
\left(b_\beta(W) - \Psi_1(P), (b_\beta(W) - \Psi_1(P))^2 - \Psi_2(P)\right)}
\]

\section{Canonical Least Favorable Submodels}

\subsection*{Introduction}
We offer a new way to construct a targeted maximum likelihood estimator for multidimensional parameters via defining the canonical least favorable submodel (clfm) \parencite{clfm}.   A more detailed version of this section by Jonathan Levy is on \href{https://arxiv.org/abs/1811.01261}{arxiv.org}, where we provide a precise general algorithm currently implemented \parencite{gentmle2}.  For the purpose of this appendix we will just show the reader how the clfm leads very naturally to the one-step TMLE \parencite{Laan:2015ab} and can be viewed as iterative version of the one-step TMLE.  Other possible advantages of using the clfm are in Levy, 2018.  Using clfm's has yet to be explored fully.
 
\subsection{Mapping $P_{n}^{0}$ to $P_{n}^{\star}$: The Targeting Step}
Here we will provide an alternate and simple construction of the universal least favorable submodel employed in this paper via the clfm.  

\begin{defn}
We can define a canonical 1-dimensional locally least favorable submodel (clfm) of an estimate, $P_{n}^{0}$, of the true distribution as
\begin{equation}
\{P_{n, \epsilon}^0 \text{ s.t } \frac{d}{d\epsilon}P_{n}L(P_{n, \epsilon}^0)\biggr\vert_{\epsilon=0}=\Vert P_{n} D^{\star}(P_{n}^{0})\Vert_{2}, \epsilon \in [-\delta,\delta]\} 
\end{equation}

where $P_{n, \epsilon}^0 = P_n^0$ and $\Vert \cdot \Vert_{2}$ is the euclidean norm.  We consider a $d-dimensional$ parameter mapping $\Psi:\mathcal{M}\longrightarrow \mathbb{R}^{d}$.
\end{defn} 

This definition differs from the locally least favorable submodel (lfm) in van der Laan and Gruber, 2016, where the empirical mean of the loss spans the efficient influence curve.   Here we can define a clfm with only a single epsilon where as with an lfm, we use an $\epsilon$ of dimension of minimum the dimension of the parameter of interest.  

\begin{defn}
A Universal Least Favorable Submodel (ulfm) of $P_{n}^{0}$ satisfies
\[
\frac{d}{d\epsilon}P_{n}L(P_{n}^{\epsilon})=\Vert P_{n}D^{\star}(P_{n}^{\epsilon})\Vert_{2}\text{ }\forall\epsilon\in(-\delta,\delta)
\]
and naturally, $P_n^{\epsilon = 0} = P_n^0$.  
\end{defn}

We can construct the universal least favorable submodel (ulfm) in terms of the clfm if we use the difference equation $P_{n}(L(P_{n,dt}^{0})-L(P_{n}^{0}))  \approx  \Vert P_{n}D^{\star}(P_{n}^{0})\Vert_{2}dt$, where $P_{n}^{dt} = P_{n, dt}^{0}$ is an element of the clfm of $P_n^0$. More generally, we can map any partition $t=m\times dt$ for an arbitrarily small, $dt$, to an equation $P_n(L(P_{n}^{t+dt})-L(P_{n}^{t}))  \approx  \Vert P_{n}D^{\star}(P_{n}^{t})\Vert_{2}dt$, where $P_{n}^{t+dt}$ is an element of the clfm of $P_{n}^{t}$. We therefore can recursively define the integral equation: $P_{n}(L(P_{n}^{\epsilon})-L(P_{n}^{0}))=\int_{0}^{\epsilon}\Vert P_{n}D^{\star}(P_{n}^{t})\Vert_{2}dt$ and  $P_n^\epsilon$ will thusly be an element of the ulfm of $P_n^0$.  For log likelihood loss, which is valid for both continuous outcome scaled between 0 and 1 as well as binary outcomes, an analytic formula for a ulfm of distribution with density, $p$, is therefore defined by the density $p_\epsilon=p\times exp(\int_0^\epsilon \Vert D^*(P^t) \Vert_2 dt)$ \parencite{Laan:2015ab} where $P^{t+dt}$ is an element of the clfm of $P^{t}$. \\

In applying the one-step TMLE, when the empirical loss is minimized at a given $\epsilon$, we will have solved, $\Vert P_{n}D^{\star}(P_{n}^{\epsilon})\Vert_{2}=0$. Therefore, the loss is decreased and all influence curve equations are solved simultaneously with a single $\epsilon$ in one step. Specifically, $P_{n}D^{\star}_{j}(P_{n}^{\star})=0$ for all $j$. Thus $P_{n}^{\star}=P_{n}^{\epsilon}$ and we have defined the required TMLE mapping.   

\section{An Easy Implementation of CV-TMLE}
\subsection*{Introduction}
The original formulation and theoretical results of cross-validated targeted maximum likelihood estimators, CV-TMLE \parencite{Zheng:2010aa}, leads to an algorithm for the CV-TMLE that generally requires 10 targeting steps for each of 10 validation folds for each iteration in an iterative targeted maximum likelihood estimators or TMLE \parencite{Laan:2006aa}.  Such can be done in one regression, which solves the efficient influence curve equation averaged over the validation folds.  However, in this pooled regression, we must keep track of the means used in each fold, making the process different than a regular TMLE, once the initial predictions have been formed.  The formulation of the CV-TMLE here-in leads to a simpler implementation of the targeting step in that the targeting step can be applied identically as for a regular TMLE once the initial estimates for each validation fold have been computed. The CV-TMLE as discussed here is currently implemented in the R software package of tlverse \parencite{sl3}. \\ 

\subsection{CV-TMLE Definition for General Estimation Problem}
We refer the reader to the following sources \parencite{Laan:2015aa,Laan:2015ab, Laan:2006aa, Laan:2011aa} for a more detailed look at the theory of TMLE and Zheng and van der Laan, 2010 for theory regarding CV-TMLE.  We consider iid data of the form $O\sim P \in \mathcal{M}$, nonparametric or semiparametric model and parameter mapping 

\[
\Psi(Q(\cdot)):\mathcal{M}\longrightarrow \mathbb{R}^d 
\]

Where $Q(P)$ is a model upon which the parameter depends.  If we consider $O=(W,A,Y)$ with outcome, $Y$, and treatment and covariates, $A$ and $W$, then the outcome model $\bar{Q}(A,W) = E_P[Y \mid A, W]$ and distribution of $W$, $Q_W$, would define $Q(P)$. We consider the canonical least favorable submodel (reference myself) of model estimate $\hat{Q}(P_n)$ defined with one-dimensional $\epsilon$:

\[
\frac{d}{d\epsilon}L\left(\hat{Q}(P_n)(\epsilon)\right)\biggr\vert_{\epsilon=0} = \Vert D^*\left(\hat{Q}(P_n), \hat{g}(P_n)\right) \Vert_2
\]

This definition coincides with the least favorable submodel if the $d=1$ because in that case we will have
\[
\langle \frac{d}{d\epsilon}L\left(\hat{Q}(P_n)(\epsilon)\right)\biggr\vert_{\epsilon=0} \rangle  \supset \langle D^*\left(\hat{Q}(P_n), \hat{g}(P_n)\right) \rangle
\]

where the above $\Vert \cdot \Vert_2$ is the euclidean norm.  We then define a mapping $B_n \in {0,1}^n$ to be a random split of ${1,..,n}$.  The training set is defined as $\mathcal{T}=\{i : B_n(i)=0\}$ and the validation set, $\mathcal{V}=\{i : B_n(i)=1\}$.  As in Zheng 2010, $P_{n,B_n}^0$ and $P_{n,B_n}^1$ and the empirical distributions over $\mathcal{T}$ and $\mathcal{V}$ respectively. \\

The CV-TMLE estimator as in Zheng and van der Laan, 2010 is defined as
\[
\Psi^{k_n}(P_n) = E_{B_n} \Psi \left(\hat{Q}(P_{n,B_n}^0)(\overset{\rightarrow}{\epsilon_n}^{k_n})\right)
\]

where $\Psi  \left(\hat{Q}(P_{n,B_n}^0)(\overset{\rightarrow}{\epsilon_n}^{k_n})\right)$ is the plug-in estimator (usually an average of the plugged-in model over the validation set).  $\overset{\rightarrow}{\epsilon_n}^{k_n}$ denotes the kth iteration of fluctuation parameters, where $k$ could always be 1 if we use the one-step TMLE \parencite{Laan:2015ab}.  

\subsection{Illustrative Example, VTE}
We will now go through the CV-TMLE algorithm for the VTE, variance of treatment effect.  Here, we notice that we never target the distribution of $W$, but rather use the unbiased estimator, the empirical distribution.  This is discussed in Zheng and van der Laan, 2010 so refer the reader there for more detail as to why this is often the case.  In short, the component of the efficient influence curve in the tangent space of mean 0 functions of $W$ \parencite{Vaart:2000aa} is given by $D^*_W(P) =  (b(P)(W) - E_P b(P)(W))^2$ where $b(P)(W) = E_P[Y \mid A = 1, W] - E_P[Y \mid A = 0, W]$.  For any approximation to this function, its empirical mean will automatically be zero. We denote the following to avoid heavy notation:
\[
\bar{Q}^k_{B_n} = \hat{Q}(P_{n,B_n}^0)(\overset{\rightarrow}{\epsilon_n}^{k})
\]
is the approximation of the outcome model at the kth iteration.  This fit is entirely dependent on the training set $P_{n,B_n}^0$ observations and the fluctuations to the model, performed on the corresponding validation set. 
\[
\bar{Q}^k_{1,B_n} = \hat{Q}(P_n)(\overset{\rightarrow}{\epsilon_n}^{k})
\]
is the approximation of the outcome model at the kth iteration.  We will see it actually depends on $P_{n,B_n}^0$ and $P_{n,B_n}^1 \hat{B}(P_{n,B_n}^0)(\overset{\rightarrow}{\epsilon_n}^{k})$, and hence the entire empirical draw of the data. 

\begin{eqnarray*}
\hat{b}^k_{B_n}(W) &=& \bar{Q}^k_{B_n}(1,W) - \bar{Q}^k_{B_n}(0,W)\\
\hat{b}^k_{1,B_n}(W) &=& \bar{Q}^k_{1,B_n}(1,W) - \bar{Q}^k_{1,B_n}(0,W)\\
\hat{g}_{B_n}(A \mid W) &=& \hat{g}(P_{n,B_n}^0)(A \mid W) 
\end{eqnarray*}

\begin{itemize}
\item
STEP 1: Initial estimates

For each split, $B_{n}$ as in standard 10-fold cross-validation, we use an ensemble learning package such as sl3 \parencite{sl3} or SuperLearner \parencite{Eric-Polley:2017aa} to fit a model on the training set, denoting the model as $P_{n, B_{n}}^0$.  In this case we will fit relevant factors of the likelihood, such as the propensity score and outcome model, but not the distribution of covariates, $W$.  For those, we use the empirical distribution as an unbiased estimator and will not target it.   Denote the initial fit of the $E_P[Y \mid A, W]$, which we denote $\bar{Q}^0_{B_n}$.  For both procedures the initial fits are all the same.  

\item
STEP 2: Check Tolerance

For each fold evaluate the so-called clever covariate: 

$H_{B_n}^k(A,W) = 2(\hat{b}^k_{B_n}(W) - P_{n,B_n}^1 \hat{b}^k_{B_n}) \frac{2A - 1}{\hat{g}_{B_n}(A \mid W)}$

and the influence curve approximation

 \[
D^*_{k, B_n}(O) = H_{B_n}^k(A,W)(Y -\bar{Q}^k_{B_n}(A,W))
\]

Our proposed procedure would do

$H_{1,B_n}^k(A,W) = 2(\hat{b}^k_{1,B_n}(W) - E_{B_n}P_{n,B_n}^1 \hat{b}^k_{1,B_n}) \frac{2A - 1}{\hat{g}_{B_n}(A \mid W)} $

and the alternate influence curve approximation

 \[
D_{k, B_n}(O)= H_{1,B_n}^k(A,W)(Y -\bar{Q}^k_{1,B_n}(A,W))
\]

Thus, in our procedure we need not keep track of the folds since the average within the clever covariate is merely taken over the entire sample.  Thus the process is identical to a TMLE once the initial estimates are made.  We just stack them on top of eachother and act is if it is all one initial fit as with the regular TMLE. \\

We then compute the influence curve approximation for each fold and take the sample mean.  Since the $T_W$ component, as stated above always has empirical average 0, we only need to take the mean of the component of the influence curve approximation in the tangent space, $T_Y =$ mean 0 functions of $Y \mid A, W$, which have finite variance \parencite{Vaart:2000aa}.  We then check if the mean of the influence curve is below the tolerance level, $\hat{\sigma}/n$ where $\hat{\sigma}$ is the sample standard deviation of the above influence curve computations.  This assures we stop the process when the bias is second order as any more fluctuations beyond that point are not helpful.  If we are below the tolerance we go to step 4.  Otherwise we continue onward.  
 
 \item
STEP 3: Targeting Step: Run a pooled logistic regression over all the folds with model: 

\[
Y = expit(logit \left(\bar{Q}^k_{B_n}(A,W) + \epsilon_n^k H(\bar{Q}^k_{B_n}(A \mid W) \right)
\]  

That is, a model which suppresses the intercept and uses and the initial predictions as the offset.  This is identical to our method, except we would use the slightly different clever covariate as stated above.  

Update all the predictions to form $\bar{Q}^{k+1}_{B_n}(A,W)$ or, as with our method $\bar{Q}^{k+1}_{1,B_n}(A,W)$.  

\item
STEP 4: Compute the estimate and CI:

\[
\Psi^{k_n}(P_n) = E_{B_n} \Psi \left(\hat{Q}(P_{n,B_n}^0)(\overset{\rightarrow}{\epsilon_n}^{k_n})\right)
\]

and estimate the standard error via the standard deviation of the influence curve in step 3 divided by root n, which we will just call $\hat{\sigma}/\sqrt{n}$ and form the confidence bands

\[
\Psi^{k_n}(P_n) \pm z_\alpha \hat{\sigma}/\sqrt{n}
\]

where $z_\alpha$ is the $1-\alpha/2$ normal quantile.  This entails computing the parameter separately per validation set before averaging the 10 estimates, i.e., compute the sample variance over the validation set for  $\hat{b}^k_{B_n}$, getting 10 estimates and then average them.  In our procedure we just have a list of n values of $\hat{b}^k_{1,B_n}$ and compute the sample variance over the entire sample.  

\end{itemize}

Thus we can see our procedure simplifies the targeting and, like the original formulation, solves the efficient influence curve equation, i.e. $E_{B_n}P_{n,B_n}^1 \hat{b}^k_{B_n}D^*_k(O)$ and $E_{B_n} P_{n,B_n}^1 \hat{b}^k_{1,B_n}D_k(O) \approx 0$, leading to the second order expansion as given in section 2.2. 

\subsection{Donsker Condition}
In the original formulation of the CV-TMLE, we view the estimator as 10 plug-in estimators.  To compute each of the 10 estimators, the targeting step is performed on the validation set.  Since we can therefore condition on the training set from which the initial estimate is formed, we essentially have a fixed functions $\bar{Q}^0_{B_n}$ and $\hat{g}_{B_n}$, which we are fluctuating on the validation set with a one-dimensional parametric submodel.  Thus the entropy is very low for the class of functions containing $\bar{Q}^k_{B_n}$ in our above algorithm.  With our procedure the entropy is a little bigger in that the function, $\bar{Q}^k_{1,B_n}$, can be viewed as fixed, yet depending on an average over all validation sets (therefore very slightly inbred before the targeting step) as well as the fluctuation parameter, $\epsilon$, determined by the validation set.  The influence curve approximation, $D_{k, B_n}$, defined above, will thus have similarly low entropy as if we allowed another parameter in the parametric submodel.\\

Consider the following, which we pull out of Zheng and van der Laan, 2010, for the convenience of the reader.  

\begin{defn}
For a class of function, $\mathcal{F}$, whose elements are functions, $f$, that map observed data, $O$, to a real number, we define the entropy integral:
\[
Entro(\mathcal{F}) = \int_0^{\infty} \sqrt{\underset{Q}{\text{log sup}}N\left(\epsilon, \Vert F \Vert_{Q,2}, \mathcal{F}, L^2(Q) \right) d\epsilon}
\]
where $N\left(\epsilon, \mathcal{F}, L^2(Q) \right)$ is the covering number for $\mathcal{F}$, defined by the minimum number of balls of radius $\epsilon$ under the $L^2(Q)$ norm to cover $\mathcal{F}$.  $F$ is defined as the envelope of $\mathcal{F}$ or a function such that $\vert f \vert \leq F$ for all $f \in \mathcal{F}$.  
\end{defn}

Consider the following lemma (lemma 2.14.1 in ref van der Vaart and Wellner, 1996) \parencite{Vaart:1996aa}

\begin{lem}
Let $\mathcal{F}$ denote a class of measurable functions of $O$.  Let $G_n = \sqrt{n}(P_n - P_0)$.  Then 

\[
E(sup_{f\in \mathcal{F}} G_n f) \leq Entro(\mathcal{F}) \sqrt{P_0 F^2}
\]
\end{lem}

This lemma then yields the following results in Zheng and van der Laan, 2010.  Consider $\overset{\rightarrow k_n}{\epsilon_n}$, a sequence of $\epsilon_n^1,...,\epsilon^{k_0}_n$ that are the fluctuation parameters dependent on the draw from the data.  In the lemma below we assume the $k_0$ steps of a parametric fluctuation parameters converge in probability to a sequence of length $k_0$, a very weak assumption, the same as the estimated parameters of a parametric model converging to the truth in probability. 
NOTE: for the one-step TMLE \parencite{Laan:2015ab} $k_n = k_0 = 1$ so the notation simplifies a bit.   

\begin{lem}
Suppose $\Vert \overset{\rightarrow k_n}{\epsilon}- \overset{\rightarrow k_0}{\epsilon}\Vert \overset{P}{\rightarrow}0$.  For each sample split of $B_n$, we consider a class of measurable functions of $O$:
\[
\mathcal{F}\left( P_{n,B_n}^0 \right) = \left\{f_{\overset{\rightarrow}{\epsilon}} \left( P_{n,B_n}^0 \right)= f\left( \overset{\rightarrow}{\epsilon}, P_{n,B_n}^0 \right) - f\left( \overset{\rightarrow}{\epsilon_0} P_{n,B_n}^0 \right):\overset{\rightarrow}{\epsilon} \right\}
\]
where the index set contains $\epsilon_n$ with probability tending to 1.  For a deterministic sequence $\delta_n\rightarrow 0$, define subclasses

\[
\mathcal{F}_{\delta_n}\left( P_{n,B_n}^0 \right) = \left\{f_{\overset{\rightarrow}{\epsilon}} \in \mathcal{F}\left( P_{n,B_n}^0 \right) : \Vert \overset{\rightarrow}{\epsilon} - \overset{\rightarrow}{\epsilon_0} \Vert < \delta_n \right\}
\]
If for deterministic sequence $\delta_n\rightarrow 0 $ we have
\[
E\left\{ Entro(\mathcal{F}_{\delta_n}\left( P_{n,B_n}^0 \right)) \sqrt{P_0 F(\delta_n, P_{n,B_n}^0)^2} \right\} \rightarrow 0 \text{ as } n\rightarrow 0
\]
where $F(\delta_n, P_{n,B_n}^0)$ is the envelope of $\mathcal{F}_{\delta_n}\left( P_{n,B_n}^0 \right)$, then
\[
\sqrt{n}(P_{n,B_n}^1 - P_0) \left\{f(\overset{\rightarrow}{\epsilon_n}, P_{n,B_n}^0) - f(\overset{\rightarrow}{\epsilon_0}, P_0) \right\} = o_P(1)
\]
\end{lem}

We note to the reader that we keep lemma 3.2 identical to what was in Zheng and van der Laan, 2010, except we do not condition solely on $P_{n,B_n}^0$ when defining $\mathcal{F}\left( P_{n,B_n}^0 \right)$.  Such does not at all affect the truth of the lemma. 

\subsubsection{Remainder Term}

Our estimate minus the truth is, using notation in Zheng and van der Laan, 2010, where $\overset{\rightarrow}{k_n},$ indicates the $k_n$ iteration, we have

\[
\Psi^{k_n}(P_n) = E_{B_n} \hat{\Psi}_{B_n} (P_{n})
\]

The second order remainder, $R_2(\cdot)$, can be written:

\begin{eqnarray*}
\Psi^{k_n}(P_n) - \Psi(P_0) &=& E_{B_n}(P_{n,B_n}^1 - P_0)D_{\overset{\rightarrow}{k_n}, B_n} + R_2(P_n, P_0)\\
&=& -E_{B_n}P_0 D_{\overset{\rightarrow}{k_n}, B_n} + R_2(P_n, P_0)
\end{eqnarray*}

Assuming the remainder is $o_P(1/\sqrt{n})$, we then get that  
\[
\Psi^{k_n}(P_n) - \Psi(P_0) = E_{B_n}(P_{n,B_n}^1 - P_0 )D_{\overset{\rightarrow}{k_n}, B_n} + o_P(1\sqrt{n}) =-E_{B_n}P_0 D_{\overset{\rightarrow}{k_n}, B_n} + R_2(P_n, P_0)
\]

since our procedure solves $E_{B_n} P_{n,B_n}^1 D_{\overset{\rightarrow}{k_n}, B_n} = 0$.  As discussed, we can quite easily satisfy lemma 3.2 for the function class containing $D_{\overset{\rightarrow}{k_n}, B_n}$.  Again, assuming the remainder is $o_P(1/\sqrt{n})$ our estimator is asymptotically efficient if $D_{\overset{\rightarrow}{k_n}, B_n}$ converges to the true influence curve in $L^2(P_0)$ \parencite{Laan:2006aa}.  For blip variance the remainder term conditions are no more strict than for the original formulation of the CV-TMLE. 

\subsection{Conclusion}
This slight adjustment to the CV-TMLE algorithm is easier to implement and retains the same theoretical properties, as shown in our example here.  It remains to be more formally generalized to include a class of TMLE's for which it is valid but the example used here-in gives the reader sufficient intuition to understand when such can be done.  For one, it is obvious if any polynomial factor of a mean (assuming the mean converges) appears as a factor in the clever covariate, then the entropy will be similarly small, so this procedure covers many examples one might find in practice.  The procedure overlaps exactly with the originally formulated CV-TMLE with many common parameters where the clever covariates contain no empirical means.  It is a subject for future research whether this procedure has any advantages in finite samples, such as in the case of simultaneously estimating the ATE, which is then used as the centering in the VTE computation.  Such appears to be perhaps more sensible but simulations have shown no appreciable difference in performance for VTE.

\newpage
\printbibliography
\clearpage

\end{document}